\documentclass[11pt]{article}

\usepackage[fleqn,leqno]{amsmath}
\usepackage{amsfonts}
\usepackage{amssymb}
\usepackage{amsthm}

\usepackage[sc]{mathpazo}
\usepackage[scaled=0.87]{berasans}
\usepackage[scaled=0.87]{beramono}
\usepackage[T1]{fontenc}
\usepackage[protrusion=true,expansion]{microtype}
\usepackage[scr=esstix,cal=boondox]{mathalfa}

\AtBeginDocument{\let\phi\varphi}
\AtBeginDocument{\let\epsilon\varepsilon}

\usepackage{geometry}
\usepackage[longnamesfirst]{natbib}

\usepackage[usenames,dvipsnames]{xcolor}
\definecolor{darkblue}{rgb}{0.0,0.0,0.39}

\usepackage{graphicx}

\usepackage[pagebackref,pdfusetitle]{hyperref}
\hypersetup{
  linktocpage = true,
  colorlinks = true,
  linkcolor = darkblue,
  urlcolor = darkblue,
  anchorcolor = darkblue,
  citecolor = darkblue,
  pdfcreator={LaTeX}
}

\usepackage{authblk}
\makeatletter
\newif\if@econjel
\@econjelfalse
\def\@jelclass{}
\newif\if@econkeywords
\@econkeywordsfalse
\def\@keywords{}
\newif\if@econack
\@econackfalse
\def\@acknowledgments{}
\newif\if@econabstract
\@econabstractfalse

\newcommand{\acknowledgments}[1]{%
  \@econacktrue
  \gdef\@acknowledgments{#1}}

\newcommand{\keywords}[1]{%
  \@econkeywordstrue
  \gdef\@keywords{#1}
  \hypersetup{ pdfkeywords = {#1}}}

\newcommand{\jelclass}[1]{%
  \@econjeltrue
  \gdef\@jelclass{#1}
  \hypersetup{ pdfsubject = {JEL Classification: #1}}}

\renewcommand{\abstract}[1]{%
  \@econabstracttrue
  \gdef\@abstract{#1}}

\newif\if@econpdfauth
\@econpdfauthfalse
\newcommand{\pdfauthor}[1]{%
  \@econpdfauthtrue
  \def\@pdfauthor{#1}}

\newcommand{\orig@maketitle}{}%
\let\orig@maketitle\maketitle
\def\maketitle{%
  \hypersetup{ pdftitle={\@title} }%
  \if@econpdfauth
    \hypersetup{ pdfauthor={\@pdfauthor} }%
  \else
    \hypersetup{ pdfauthor={\@author} }%
  \fi
  \orig@maketitle
}

\def\@maketitle{%
  \newpage
  \begin{center}%
  \let \footnote \thanks
    \if@econack
      {\huge \@title\thanks{\@acknowledgments} \par}%
    \else%
      {\huge \@title \par}%
    \fi
    \vskip 2em%
    {\large
      \lineskip .5em%
      \begin{tabular}[t]{c}%
        \@author
      \end{tabular}\par}%
    {\@date}%
  \end{center}%
  \par
  \if@econabstract
  \vskip 1.5em
  \begin{center}
    \begin{minipage}{0.9\textwidth}
      \small
      {\noindent {\bfseries Abstract.}\hspace{0.5em} \@abstract}\par%
      \if@econkeywords
        \medskip%
        \noindent%
        \textbf{Keywords:} \@keywords.\par%
      \fi
      \if@econjel
        \medskip%
        \noindent%
        \textbf{JEL Classification:} \@jelclass.
      \fi
    \end{minipage}
  \end{center}
  \fi
  \vskip 0.5em}%

\long\def\@makecaption#1#2{%
  \vskip\abovecaptionskip
  \sbox\@tempboxa{\textsc{#1}. #2}%
  \ifdim \wd\@tempboxa >\hsize
    \textsc{#1}. #2\par
  \else
    \global \@minipagefalse
    \hb@xt@\hsize{\hfil\box\@tempboxa\hfil}%
  \fi
  \vskip\belowcaptionskip}

\renewcommand{\@seccntformat}[1]{\begingroup\csname the#1\endcsname. \endgroup}

\renewcommand\section{\@startsection {section}{1}{\z@}%
  {-3.5ex \@plus -1ex \@minus -.2ex}%
  {2.3ex \@plus.2ex}%
  {\normalbaselines\normalfont\large\bfseries\raggedright}}

\renewcommand\subsection{\@startsection{subsection}{2}{\z@}%
  {-3.25ex\@plus -1ex \@minus -.2ex}%
  {1.5ex \@plus .2ex}%
  {\normalbaselines\normalfont\normalsize\itshape\raggedright}}

\renewcommand\subsubsection{\@startsection{subsubsection}{3}{\z@}%
  {-3.25ex\@plus -1ex \@minus -.2ex}%
  {1.5ex \@plus .2ex}%
  {\normalbaselines \normalfont\normalsize\itshape\raggedright}}

\renewcommand\paragraph{\@startsection{paragraph}{4}{\z@}%
  {3.25ex \@plus 1ex \@minus.2ex}%
  {-1em}%
  {\normalbaselines\normalfont\small\itshape\raggedright}}

\renewcommand\subparagraph{\@startsection{subparagraph}{5}{\parindent}%
  {3.25ex \@plus 1ex \@minus .2ex}%
  {-1em}%
  {\normalbaselines\normalfont\small\itshape\raggedright}}

\renewcommand*{\backref}[1]{}
\makeatother

\theoremstyle{plain}
\newtheorem{theorem}{Theorem}
\newtheorem{lemma}{Lemma}
\theoremstyle{definition}
\newtheorem*{definition}{Definition}
\newtheorem{assumption}{Assumption}
\theoremstyle{remark}
\newtheorem*{remark}{Remark}

\DeclareMathOperator{\e}{\mathrm{e}}

\DeclareRobustCommand{\R}[0]{\mathbb{R}}
\DeclareRobustCommand{\Z}[0]{\mathbb{Z}}

\DeclareMathOperator{\E}{\mathrm{E}}
\DeclareMathOperator{\1}{\mathrm{1}}
\newcommand{\abs}[1]{\left\lvert#1\right\rvert}
\newcommand{\floor}[1]{\left\lfloor#1\right\rfloor}

\DeclareMathOperator{\diag}{diag}
\DeclareMathOperator{\vectorize}{vec}

\DeclareMathOperator{\Var}{\mathrm{Var}}
\DeclareMathOperator{\Cov}{\mathrm{Cov}}
\DeclareMathOperator{\var}{\Var}
\DeclareMathOperator{\cov}{\Cov}

\DeclareMathOperator{\Normal}{\mathrm{N}}
\DeclareMathOperator{\Exponential}{\mathrm{Expo}}

\usepackage{enumerate}
\usepackage[capposition=bottom]{floatrow}

\usepackage{setspace}
\theoremstyle{example}
\newtheorem*{ex:renewal}{Renewal Example}
\newtheorem*{ex:entry}{$\mathbf{2 \times 2}$ Entry Example}
\newtheorem*{ex:ladder}{Quality Ladder Example}

\def\Nplayers{N}
\def\proofapx{}

\begin{document}

\title{Identification and Estimation of Continuous-Time Dynamic Discrete Choice Games}

\author{\uppercase{Jason R. Blevins} \\ The Ohio State University}

\pdfauthor{Jason R. Blevins}


\date{November 4, 2025}

\jelclass{%
  C13, 
  C35, 
  C62, 
  C73} 

\abstract{This paper considers the theoretical, computational, and econometric properties of continuous time dynamic discrete choice games with stochastically sequential moves, introduced by \citet*{abbe-2016}.
We consider identification of the rate of move arrivals, which was assumed to be known in previous work, as well as a generalized version with heterogeneous move arrival rates.
We re-establish conditions for existence of a Markov perfect equilibrium in the generalized model and consider identification of the model primitives with only discrete time data sampled at fixed intervals.
Three foundational example models are considered: a single agent renewal model, a dynamic entry and exit model, and a quality ladder model.
Through these examples we examine the computational and statistical properties of estimators via Monte Carlo experiments and an empirical example using data from \cite{rust87optimal}.
The experiments show how parameter estimates behave when moving from continuous time data to discrete time data of decreasing frequency and the computational feasibility as the number of firms grows.
The empirical example highlights the impact of allowing decision rates to vary.}

\acknowledgments{%
  Replication package available at \url{https://github.com/jrblevin/ctgames-qe}.
  I am grateful for useful comments from the editor and three referees, as well as discussions with seminar participants at
  Columbia University,
  Indiana University,
  Northwestern University,
  Stony Brook University,
  Texas A\&M University,
  the University of British Columbia,
  the University of Chicago,
  the University of Iowa,
  the University of Michigan,
  and the University of Montreal,
  and conference attendees
  at the 2013 Meeting of the Midwest Econometrics Group,
  the 2014 Meeting of the Midwest Economics Association,
  the 2014 University of Calgary Empirical Microeconomics Workshop,
  the 2015 Econometric Society World Congress
  and the 2019 International Association for Applied Econometrics conference.
  This work builds on earlier work with
  Peter Arcidiacono,
  Patrick Bayer,
  and
  Paul Ellickson
  and benefited tremendously from our many discussions together.
}

\maketitle
\newpage

\section{Introduction}
\label{sec:intro}

This paper studies continuous-time econometric models of dynamic
discrete choice games.
Work on continuous-time dynamic games by \cite{doraszelski-judd-2012},
\citet*{abbe-2016} (henceforth ABBE), and others was motivated by their
ability to allow researchers to compute and estimate more realistic,
large-scale games and to carry out complex counterfactual
policy experiments which were previously infeasible due to
computational limitations.

Given the practical and conceptual benefits of continuous-time models,
this paper considers identification and
estimation of the rate of move arrivals in the original ABBE model,
where it was assumed known.
We also consider identification in a generalized
model with additional heterogeneity through
firm- and state-specific move arrival rates.
We demonstrate these specifications via
three canonical models: a single agent renewal model, a dynamic
entry and exit model, and a quality ladder model.
We carry out an empirical illustration using the original data of
\cite{rust87optimal} to estimate a continuous-time model with
heterogeneous move arrival rates across states and compare with
the restricted form of the original ABBE model.
Based on the estimates, we conduct Monte Carlo experiments
to investigate the effects of estimating the model using
discrete-time data of varying frequencies.
Using the quality ladder model, we conduct
Monte Carlo experiments to examine the computational
feasibility as the number of firms grows.

For many economic models there is no natural, fixed time interval
at which agents make decisions.
Despite this, it is standard practice for applied researchers to
calibrate the decision interval in their empirical model to be
equal to the sampling interval of their data.
However, continuous-time modeling offers more flexibility
by allowing agents to make decisions asynchronously at
stochastic points in time.
This approach eliminates simultaneous moves and introduces
sequential decision-making, which better aligns with real-world
scenarios in many cases and leads to significant computational
advantages.

Another advantage of continuous-time modeling is the
reduction of multiple equilibria by eliminating simultaneous
moves.
While it does not completely eliminate multiplicity, this simplifies
estimation and the ability to conduct meaningful counterfactual
simulations.
Therefore, another benefit of allowing for heterogeneity in move
arrival rates is to reduce symmetry in the model and thereby remove
another source of multiplicity.

Modeling economic processes in continuous time dates back
several decades to work in time series econometrics by
\cite{phillips-1972, phillips-1973},
\cite{sims-1971},
\cite{geweke-1978},
and
\cite{geweke-marshall-zarkin-1986}
and work on longitudinal models by \cite{heckman-singer-1986}.
Despite this early work on continuous-time models, discrete-time
models became the de facto standard for dynamic discrete choice and
now have a long, successful history in structural applied microeconometrics
starting with the
work of \cite{gotz80estimation}, \cite{miller-1984}, \cite{pakes-1986},
\cite{rust87optimal}, and \cite{wolpin84estimable}.
A recent series of papers
\citep{aguirregabiria-mira-2007, bajari-benkard-levin-2007,
  pakes-ostrovsky-berry-2007, pesendorfer08asymptotic} have shown how
to extend two-step estimation techniques, originally developed by
\cite{hotz93conditional} and \cite{hotz94simulation}, to more
complex multi-agent settings.
The computation of multi-agent models remains formidable, despite a growing
number of methods for solving for equilibria \citep{pakes94computing,
pakes01stochastic, doraszelski-satterthwaite-2010}.

Dynamic decision problems are inherently complex and high-dimensional,
especially in strategic games where multiple players interact.
In discrete-time models, simultaneous actions introduce
further complexity, as one must calculate expectations over all
possible combinations of rivals' actions.
This exponentially increases the number of future states to
evaluate, making it infeasible to compute equilibrium in many economic
environments.
This has severely limited the scale and degree
of heterogeneity in applied work using these methods.

\cite{doraszelski-judd-2012} showed that continuous-time models
enjoy significant computational advantages.
By eliminating simultaneous moves, the model ensures that state
changes occur sequentially.
Expectations over rival actions grow linearly with the
number of players, rather than exponentially.
As a result, continuous-time models significantly reduce the computational
burden, allowing for faster and more scalable computation of equilibria.

ABBE demonstrated the empirical tractability of continuous-time games.
They developed an econometric model retaining the aforementioned
computational advantages while incorporating familiar features from
discrete-time discrete choice models.
They proposed a two-step conditional choice probability (CCP)
estimator for their model, connecting continuous-time games with
an established line of work on discrete-time dynamic games.
They showed it is feasible to estimate extremely large-scale
games and carry out counterfactuals
in those games, which would be computationally prohibitive in
a simultaneous-move, discrete-time model.
ABBE illustrated these advantages through an empirical application
analyzing Walmart's entry into the U.S. supermarket industry.

However, ABBE's identification results did not address identification of the
rate of move arrivals.
Here, we treat it as a structural parameter to identify and allow it to depend
on the model state and player identity, with certain restrictions.
These results are important empirically, to allow flexibility in the frequency of decisions
in the model as separate from the frequency of observations in the data.
As we show empirically using \cite{rust87optimal}'s original data, allowing
and estimating varying decision rates can avoid bias in other structural
parameters, such as costs, and improve our interpretation of results.

This paper builds on \cite{blevins-2017}, which addressed identification of
the reduced forms of continuous-time models using discrete-time data.
While that work considered first-order linear systems of stochastic
differential equations, we apply those results to a specific class of
finite-state Markov jump processes generated by our structural model.
We develop linear restrictions for our model that satisfy the conditions
of Theorem 1 of \cite{blevins-2017} to identify the continuous-time reduced
form of our model.
We also address the question of identifying the structural primitives of our model.

Similar continuous-time models have since been used in a growing number
of applications including
\cite{takahashi-2015} to movie theaters,
\cite{deng-mela-2018} to TV viewership and advertising,
\cite{nevskaya-albuquerque-2019} to online games,
\cite{agarwal-ashlagi-rees-somaini-waldinger-2021} to allocation of donor kidneys,
\cite{jeziorski-2022} to the U.S. radio industry,
\cite{schiraldi-smith-takahashi-2012-wp} to supermarkets in the U.K.,
\cite{lee-roberts-sweeting-2012-wp} to baseball tickets,
\cite{cosman-2014-wp} to bars in Chicago,
\cite{mazur-2017-wp} to the U.S. airline industry,
\cite{kim-2021-does} to the U.S. retail banking industry,
and
\cite{qin-vitorino-john-2022-wp} to airline networks in China.

The remainder of this paper is organized as follows.
In \autoref{sec:model}, we review a generalized version of the ABBE
model in which move arrival rates may vary by
player and state.
We establish a linear representation of the value function in terms of
CCPs as well as the existence of a Markov perfect equilibrium.
We then develop new identification results for the model in
\autoref{sec:identification}.
We use two canonical examples throughout the paper to illustrate our
results: a single agent renewal model based on \cite{rust87optimal}
and a $2 \times 2$ entry model similar to models used by
\cite{aguirregabiria-mira-2007}, \cite{pesendorfer08asymptotic},
and others.
Finally, in \autoref{sec:mc} we examine the computational and
econometric properties via a series
of Monte Carlo experiments.
\autoref{sec:conclusion} concludes.

\section{Continuous-Time Dynamic Discrete Choice Games}
\label{sec:model}

We consider infinite horizon discrete games in continuous time indexed by $t \in [0, \infty)$ with $\Nplayers$ players indexed by $i = 1, \dots, \Nplayers$.
We introduce a heterogeneous generalization of the ABBE model, where players may have different discount rates and move arrival rates across states.
After formalizing the structural model, we establish a linear representation of the value function in terms of CCPs and existence of a Markov perfect equilibrium.
We conclude the section with a comparison of discrete- and continuous-time models.

\subsection{State Space}

At any instant, the payoff-relevant market conditions can be summarized by a state vector $x \in \mathcal{X}$ with $K \equiv |\mathcal{X}| < \infty$.
Each state $x \in \mathcal{X}$ contains information about the market structure (e.g., which players are active, the quality of each player) and market conditions (e.g., demographic and geographic characteristics, input prices).

States $x$ are typically $L$-dimensional vectors in $\R^L$, with components that represent player-specific states and market characteristics.
For example, $x = (x_1, \dots, x_\Nplayers, d)$ where $x_i$ are player-specific states (e.g., incumbency status or number of stores) and $d$ is an exogenous market characteristic (e.g., population or demand level).

Because the state space is finite there is an equivalent \emph{encoded} state space representation $\mathcal{K} = \{ 1, \dots, K \}$.
Although $\mathcal{X}$ is the most natural way to interpret the state, using $\mathcal{K}$ allows us to easily vectorize payoffs, value functions, and other quantities.

\begin{ex:renewal}
\label{ex:rust}
Consider a continuous-time version of the single-agent renewal
model of \cite{rust87optimal}.
The manager of a municipal bus company faces a
dynamic decision about when to replace a bus engine.
Replacement incurs an immediate cost but resets the engine's
mileage, reducing maintenance costs and decreasing breakdown
likelihood.
A single discrete state variable $x \in \{ x_1, x_2, \dots, x_K \}$
represents the accumulated engine mileage.
In our empirical example, each $x_k$ represents a mileage
bin $[5000 \times (k-1), 5000 \times k)$, with $K = 90$ and a maximum
of 450,000 miles.
Without loss of generality, we can represent the mileage by an integer
$k \in \mathcal{K} = \{ 1, \dots, K \}$.
\end{ex:renewal}

\begin{ex:entry}
\label{ex:2x2x2}
As a second example, consider a simple model involving two firms $i \in \{ 1, 2 \}$
who sell the same good or service, each deciding whether to enter or exit a market.
Each firm has two actions $j \in \{ 0, 1 \}$.
The choice $j = 1$ is a switching action: enter the market if inactive, or exit the market if active,
while $j = 0$ represents a continuation choice to remain active or inactive, as the case may be.
Firms observe their own and each other's market activity status $x_{1k}$ and $x_{2k}$,
along with an exogenous demand state $d$ that can be high ($\text{H}$) or low ($\text{L}$).

The model is a two-firm entry game with a binary exogenous state variable.
The state vector $x_k$ contains $x_{1k}, x_{2k} \in \{ 0, 1 \}$ and $d_{k} \in \{ \text{L}, \text{H} \}$.
The state space is
\begin{equation*}
  \mathcal{X} = \{  (0, 0, \text{L}),  (1, 0, \text{L}),  (0, 1, \text{L}),  (1, 1, \text{L}),
                    (0, 0, \text{H}),  (1, 0, \text{H}),  (0, 1, \text{H}),  (1, 1, \text{H})  \}.
\end{equation*}
We can represent this in encoded form as
$\mathcal{K} = \{ 1, 2, 3, 4, 5, 6, 7, 8 \}$.
This representation will be more analytically convenient to characterize the model.
\end{ex:entry}

Although our two running examples are simple, to
better illustrate the ideas,
in \autoref{sec:examples:ladder} we introduce a third example:
a quality ladder model of oligopoly
dynamics with heterogeneous firms based on the model of
\cite{ericson95markov-perfect}.\footnote{As another example,
  \cite{ctnpl} specify a continuous-time version of the
  dynamic entry-exit model of \cite{aguirregabiria-mira-2007}.}

\subsection{Exogenous State Changes}

The state of the model can evolve over time in response to exogenous events, which we attribute to an artificial player referred to as ``nature,'' indexed by $i = 0$.
This player is responsible for state changes that cannot be attributed to the action of any other player $i > 0$ (e.g., changes in population or per capita income).
When the model is in state $k$, let $q_{kl}$ denote the hazard rate for transitions to another state $l \neq k$.
The rate $q_{kl}$ may be zero if direct transitions from $k$ to $l$ are not possible, or $q_{kl}$ may be some positive but finite value representing the hazard rate of such a transition.
Therefore, the overall rate at which the system leaves state $k$ for any other state $l \neq k$ is $\sum_{l \neq k}^K q_{kl}$.

\begin{ex:renewal}[continued]
Suppose the exogenous mileage transition process is characterized by a rate parameter $\gamma$ governing mileage increases to the next state.
This rate is constant across states for simplicity, so for all $l \neq k$ we have
\begin{equation*}
  q_{kl} = \begin{cases}
    \gamma & \text{if } l = k + 1, \\
    0 & \text{otherwise}.
  \end{cases}
\end{equation*}
\end{ex:renewal}

\begin{ex:entry}[continued]
  In the $2 \times 2$ entry model, there are two exogenous states: high demand ($d = \text{H}$) and low demand ($d = \text{L}$).
  Suppose nature switches from $\text{H}$ to $\text{L}$ at rate $\gamma_{\text{HL}}$ and back to $\text{H}$ at rate
  $\gamma_{\text{LH}}$.
  Thus, we have
  \begin{equation}
    q_{kl} = \begin{cases}
      \gamma_{\text{HL}} & \text{if } d_{k} = \text{H} \text{ and } d_{l} = \text{L}, \\
      \gamma_{\text{LH}} & \text{if } d_{k} = \text{L} \text{ and } d_{l} = \text{H}, \\
      0 & \text{otherwise}.
    \end{cases}
  \end{equation}
\end{ex:entry}

\subsection{Decisions \& Endogenous State Changes}

As in discrete time games, the players in our model can take actions that influence the evolution of the state vector.
Each player has $J$ actions represented by the choice set $\mathcal{J} = \{ 0, 1, 2, \dots, J-1 \}$.
When the model is in state $k$, the holding time until the next move by player $i$ is exponentially distributed with rate parameter $\lambda_{ik}$.
In other words, decision times for player $i$ in state $k$ occur according to a Poisson process with rate $\lambda_{ik}$.
We assume these processes are independent across players and the rates $\lambda_{ik}$ are finite for all $i$ and $k$, reflecting the fact that monitoring the state and making decisions is costly, making continuous monitoring infeasible.

In ABBE and previous applications of this framework, the rate of decisions was
assumed to be known by the researcher and to be constant across players and
states.
For example, $\lambda_{ik} = 1$ would correspond to a decision on average once
per time unit.
In this paper, we consider the rates $\lambda_{ik}$ to be structural parameters
to be estimated.
Additional identifying restrictions will be required and therefore the
specification of these rates will be important.

Let $h_{ijk}$ denote the rate at which player $i$ takes action $j$ in state $k$, with the overall decision rate in state $k$ satisfying $\sum_{j=0}^{J-1} h_{ijk} = \lambda_{ik}$.
The choice-specific hazards are determined endogenously in equilibrium as discussed in detail in the following sections.
When player $i$ chooses action $j$, the state jumps immediately and deterministically from $k$ to the continuation state denoted by $l(i,j,k)$.

The assumption of deterministic state changes easily accommodates decisions such as market entry, price adjustments, or construction of a new store, which are direct and certain.
Our framework can also accommodate stochastic outcomes if both the decision and outcome are discrete, observable, and encoded in the state vector.
The uncertainty of the outcome can be attributed to ``nature'' and the rates of state changes that result would be parameters of the exogenous state transition process discussed in the previous section.\footnote{Consider an example of R\&D investment with an uncertain success rate.
If the firm's R\&D investment is an observable choice and encoded in the state vector (say, $j \in \{ 0, 1 \}$ switches the firm's R\&D state $x_{i,\text{r}} \in \{ 0, 1 \}$) and if the success is observable (say, a new product is either developed or not, $x_{i,\text{p}} \in \{ 0, 1 \}$), then our model allows this by treating the new product development as an uncertain outcome determined by nature, following the R\&D investment, with an estimable rate of success.}

In most economic models, the actions of players only affect their individual components of the overall state vector.
For example, when a new firm enters a market it may change the firm-specific activity indicator for that firm but not the level of demand in the market.
As we will discuss in more detail below, this leads to sparsity of the continuous time model and helps with identification.

\begin{ex:renewal}[continued]
Since there is a single agent ($\Nplayers = 1$),
we drop the subscript $i$ from the notation for this example.
Suppose the manager decides whether to replace a bus engine ($j=1$) or continue without replacing ($j = 0$).
Hence, $\mathcal{J} = \{ 0, 1 \}$.
Continuation does not change the state, but upon replacement the state resets immediately to $k = 1$, therefore
  \begin{equation*}
    l(j,k) = \begin{cases}
      k & \text{if } j = 0, \\
      1 & \text{if } j = 1.
    \end{cases}
  \end{equation*}

The agent makes decisions in each state $k$ at times determined by an
exogenous Poisson process with rate parameter $\lambda_{k}$.
This process represents the distribution of times when the manager
considers whether to replace the engine of a bus in mileage state $k$.
In a simple model, we may assume the decision rate is constant across
states: $\lambda_{k} = \lambda$.
Alternatively, we could allow that the manager evaluates buses with
higher mileage more frequently than those with lower mileage:
\begin{equation*}
  \lambda_{k} = \begin{cases}
    \lambda_{\text{L}} & \text{if } k \leq \floor{\frac{K}{2}}, \\
    \lambda_{\text{H}} & \text{otherwise}.
  \end{cases}
\end{equation*}
In this case, $\lambda_{\text{L}}$ is the rate of evaluation of
a low-mileage bus (in the lower half of states) and
$\lambda_{\text{H}}$ is the rate at which a bus with higher
mileage is monitored.

Let $h_{1k}$ denote the reduced form hazard of engine replacement in state $k$.
The rate of replacement $h_{1k}$ plus the rate of
continuation $h_{0k}$ in each state $k$ must be such that
$h_{1k} + h_{0k} = \lambda_{k}$.
Before discussing how these choice-specific hazards are determined optimally,
we need to first formalize the transition dynamics of the
state vector and introduce the payoff functions of the players.
\end{ex:renewal}

\begin{remark}
It is important to note that the endogenous hazards of specific
actions $h_{jk}$ may vary across states regardless of whether there
is heterogeneity in move arrival rates $\lambda_k$.
In practice one could assume the overall rate of decisions is constant across states:
$\lambda_{\text{L}} = \lambda_{\text{H}} = \lambda$.
This would imply the rate of (unobservable) non-replacement is $h_{0k} = \lambda - h_{1k}$.
Even in this case, with a constant overall rate of decisions, the rates of replacement and
non-replacement are endogenous and vary across states.
This is similar to the case of discrete time models, where the sum of
CCPs is necessarily constant and equal to
one while the individual choice probabilities vary across states.
The continuous time model allows another degree of flexibility in that
the rate of move arrivals can be different from one.
Heterogeneity in $\lambda_{ik}$ allows for even more flexible structures.
\end{remark}

\begin{ex:entry}[continued]
  In the $2 \times 2$ entry model, each firm $i$ makes decisions about
  entering or exiting the market in each state $k$ at rates $\lambda_{ik}$.
  We may believe that firms are heterogeneous, monitoring the market at different rates, but at possibly the same rate across states: $\lambda_{ik} = \lambda_i$.
  Alternatively, one could specify a model where firms can monitor the market
  more (or less) closely when demand is high ($d = \text{H}$) than when
  demand is low ($d = \text{L}$):
  \begin{equation*}
    \lambda_{ik} = \begin{cases}
      \lambda_{\text{L}} & \text{if } d = \text{L}, \\
      \lambda_{\text{H}} & \text{otherwise}.
    \end{cases}
  \end{equation*}
  These are merely two examples.
  We consider a third possibility---a model where the move arrival
  rates depend on the endogenous decisions of the players---in
  \autoref{sec:examples:ladder}.
\end{ex:entry}

\subsection{Payoffs}

In the continuous-time setting, we distinguish between the flow
payoffs that a player receives while the model remains in state $k$,
denoted $u_{ik}$, and the instantaneous choice-specific payoffs from
making choice $j$ in state $k$ at a decision time $t$, denoted
$c_{ijk}(t)$.
The instantaneous payoffs are additively separable as
$c_{ijk}(t) = \psi_{ijk} + \varepsilon_{ijk}(t)$, where $\psi_{ijk}$
is the mean payoff and $\varepsilon_{ijk}(t)$ is a choice-specific
unobserved payoff.
Player $i$ observes the vector
$\varepsilon_{ik}(t) \equiv \left( \varepsilon_{ijk}(t), j = 0,\dots,J-1 \right)$
of choice-specific unobservables before choosing action $j$.
All players and the researcher observe the state $k$, but
only player $i$ observes $\varepsilon_{ik}(t)$.

\begin{remark}
Note that in discrete time models, because all actions and state changes resolve simultaneously, the period payoffs are written as functions of the state, the unobservables, \emph{and the actions of all players} (e.g., $u_i(a_1,\dots,a_\Nplayers,x_t,\varepsilon_{it})$).
In our continuous-time model, the payoffs resulting from competition in the product market accrue as flows $u_{ik}$ in a specific state $k$ while the choice-specific payoffs $c_{ijk}(t)$ accrue at the instant the decision is made.
\end{remark}

\begin{ex:renewal}[continued]
  In the renewal model the agent faces a dynamic, stochastic cost minimization problem where the flow utility $u_{ik}$ is the flow cost of operating a bus with mileage $k$.
  For example, if the cost of mileage is $\beta < 0$ then a parametric flow utility function could be
  $u_{ik} = \beta k$.
  No cost is paid to continue, but a cost $\mu < 0$ is paid to replace the engine:
  \begin{equation*}
    \psi_{ijk} = \begin{cases}
      0 & \text{if } j = 0, \\
      \mu & \text{if } j = 1.
    \end{cases}
  \end{equation*}
  Following any choice $j$, the agent also receives the iid shock $\varepsilon_{ijk}$ associated with that choice.
\end{ex:renewal}

\subsection{Assumptions}

Before turning to the equilibrium, we pause and collect our assumptions so far.

\begin{assumption}[Discrete States]
  \label{assp:dx}
  The state space is finite:
  $K \equiv \abs{\mathcal{X}} < \infty$.
\end{assumption}

\begin{assumption}[Discount Rates]
  \label{assp:rho}
  The discount rates $\rho_i \in (0,\infty), i = 1, \dots, \Nplayers$ are known.
\end{assumption}

\begin{assumption}[Move Arrival Times]
  \label{assp:rates}
  Move arrival times follow independent Poisson processes with rate parameters
  $\lambda_{ik}$ for each player $i = 1, \dots, \Nplayers$ and state $k = 1, \dots, K$ and
  $q_{kl}$ for exogenous state changes from each state $k$ to $l \neq k$ due to nature,
  with $0 \leq \lambda_{ik} < \infty$, $0 \leq q_{kl} < \infty$,
  and $\sum_{l \neq k} q_{kl} + \sum_m \lambda_{mk} > 0$.
\end{assumption}

\begin{assumption}[Bounded Payoffs]
  \label{assp:u}
  The flow payoffs and choice-specific payoffs satisfy $\abs{u_{ik}} < \infty$ and $\abs{\psi_{ijk}} < \infty$ for all $i = 1, \dots \Nplayers$, $j = 0, \dots, J-1$, and $k = 1, \dots, K$.
\end{assumption}

\begin{assumption}[Additive Separability]
  \label{assp:as}
  The instantaneous payoffs are additively separable as $c_{ijk}(t) = \psi_{ijk} + \varepsilon_{ijk}(t)$.
\end{assumption}

\begin{assumption}[Costless Continuation \& Distinct Actions]
  \label{assp:dn} 
  For all $i$ and $k$:
  \begin{enumerate}[(a)]
  \item $l(i,j,k) = k$ and $\psi_{ijk} = 0$ for $j = 0$,
  \item $l(i,j,k) \neq l(i,j',k)$ for all $j = 0, \dots, J-1$ and $j' \neq j$.
  \end{enumerate}
\end{assumption}

\begin{assumption}[Private Information]
  \label{assp:pi}
  The choice-specific shocks $\varepsilon_{ik}(t)$ are iid
  across players $i$, states $k$, and decision times $t$.
  The joint distribution $F_{ik}$ is known and is
  absolutely continuous with respect to Lebesgue measure,
  with finite first moments and support $\R^{J}$.
\end{assumption}

Assumptions~\ref{assp:dx}--\ref{assp:pi} are generalized counterparts of Assumptions 1--4 of ABBE that allow for player heterogeneity and state dependent rates.\footnote{Specifically, Assumption~\ref{assp:dx} is equivalent to Assumption 1 of ABBE, Assumptions~\ref{assp:rho} and \ref{assp:rates} generalize Assumptions 2(a) and 2(b--c) of ABBE, Assumption~\ref{assp:u} is equivalent to Assumptions 2(d--e) of ABBE, and Assumptions~\ref{assp:as}--\ref{assp:dn} are equivalent to Assumptions 3--4 of ABBE, and Assumption~\ref{assp:pi} generalizes Assumption 5 of ABBE.}
Assumptions~\ref{assp:dx}--\ref{assp:as} were discussed above.
Assumption~\ref{assp:dn} formalizes that $j = 0$ is a costless continuation action and that all choices are observationally distinct.
The first part of \autoref{assp:dn} requires that if an inaction decision which does not change the state, denoted $j = 0$, is included in the choice set, then the instantaneous payoff associated with that choice must be zero.\footnote{The role of the choice $j = 0$ is similar to the role of the ``outside good'' in models of demand. Because not all agents in the market are observed to purchase one of the goods in the model, their purchase is defined to be the outside good.}  This is an identifying assumption.
The second part of \autoref{assp:dn} requires actions to be meaningfully distinct in the ways they change the state.

Finally, we formalize a common distributional assumption used in applied work.
We will use this assumption in examples and results throughout the paper for its
tractability.  This assumption implies Assumption~\ref{assp:pi}.

\begin{assumption}[Type I Extreme Value Distribution]
  \label{assp:tiev}
  The choice-specific shocks $\varepsilon_{ik}(t)$ are iid
  across players $i$, choices $j$, states $k$, and decision times $t$ and are
  distributed according to the standard
  Type I extreme value distribution.
\end{assumption}

\subsection{Strategies and Best Responses}

A stationary Markov policy for player $i$ is a function $\delta_i: \mathcal{K} \times \R^{J} \to \mathcal{J}: (k, \varepsilon_{ik}) \mapsto \delta_i(k, \varepsilon_{ik})$ mapping each state $k$ and vector $\varepsilon_{ik}$ to an action.
Associated with each policy $\delta_i$ are CCPs
\begin{equation}
\label{ccps}
\Pr[ \delta_i(k, \varepsilon_{ik}) = j \mid k ].
\end{equation}
Since firm $i$'s payoffs depend on rival shocks $\varepsilon_{mjk}$ only through their choices, it is sufficient to consider beliefs in terms of CCPs.
Let $\varsigma_{im}$ denote player $i$'s beliefs about player $m$: a collection of $J \times K$ probabilities.
Let $\varsigma_{i} = (\varsigma_{i1}, \dots, \varsigma_{i,i-1}, \varsigma_{i,i+1}, \dots \varsigma_{i\Nplayers})$ denote player $i$'s beliefs about all other players.
Finally, let $V_{ik}(\varsigma_i)$ denote player $i$'s expected present value in state $k$ when behaving optimally while rivals follow strategies consistent with beliefs $\varsigma_i$.
The best response strategy for player $i$ is
\begin{equation}
  \label{eq:best_response}
  b_i(k, \varepsilon_{ik}, \varsigma_i) =
  \arg\max_{j \in \mathcal{J}} \left\{  \psi_{ijk} + \varepsilon_{ijk} + V_{i,l(i,j,k)}(\varsigma_i) \right\}.
\end{equation}
That is, at each decision time the best response function $b_i$ assigns the action that maximizes the agent's expected payoff.
The quantities on the right side are the instantaneous payoff $\psi_{ijk} + \varepsilon_{ijk}$ associated with choice $j$ plus the present discounted value of payoffs that occur in the continuation state $l(i,j,k)$ arising when player $i$ chooses action $j$ in state $k$.

\begin{remark}
  With discrete choices, the best response condition in \eqref{eq:best_response} amounts to a threshold-crossing model with an additively separable error term.
  Under Assumption~\ref{assp:tiev} the best response probabilities have a logistic functional form in terms of the value function:
  \begin{equation}
    \label{eq:ccp:logit}
    \Pr\left[ b_i(k, \varepsilon_{ik}, \varsigma_i) = j \mid k \right]
    = \frac{\exp\left(\psi_{ijk} + V_{i,l(i,j,k)}(\varsigma_{i})\right)}{\sum_{j' \in \mathcal{J}} \exp\left(\psi_{ij'k} + V_{i,l(i,j',k)}(\varsigma_{i})\right)}.
  \end{equation}
\end{remark}

\subsection{Value Function}

Given beliefs $\varsigma_i$ held by player $i$, we can define the value function (here, a $K$-vector) $V_i(\varsigma_i) = ( V_{i1}(\varsigma_i), \dots, V_{iK}(\varsigma_i) )^\top$ where the $k$-th element $V_{ik}(\varsigma_i)$ is the present discounted value of all future payoffs obtained when starting in some state $k$ and behaving optimally in future periods given beliefs $\varsigma_i$.
For a small time increment $\tau$, under Assumption~\ref{assp:rates} the probability of an event with rate $\lambda_{ik}$ occurring is $\lambda_{ik} \tau$.
Given the discount rate $\rho_i$, the discount factor for such increments is $1 / (1 + \rho_i \tau)$.
Thus, for small time increments $\tau$ the present discounted value of being in state $k$ is
\begin{multline*}
  V_{ik}(\varsigma_{i}) = \frac{1}{1 + \rho_i \tau}\left[
    u_{ik} \tau +
    \sum_{l \neq k} q_{kl} \tau V_{il}(\varsigma_{i}) +
    \sum_{m \neq i} \lambda_{mk} \tau \sum_{j=0}^{J-1} \varsigma_{imjk} V_{i,l(m,j,k)}(\varsigma_{i})
    \right. \\ \left. +
    \lambda_{ik} \tau \E \max_j \left\{ \psi_{ijk} + \varepsilon_{ijk} + V_{i,l(i,j,k)}(\varsigma_{i}) \right\} +
    \left(1 - \sum_{m=1}^\Nplayers \lambda_{mk} \tau - \sum_{l \neq k} q_{kl} \tau \right) V_{ik}(\varsigma_{i}) +
    o(\tau)
  \right].
\end{multline*}
The $o(\tau)$ term accounts for the probabilities of two or more Poisson events occurring during the small interval $\tau$, which are proportional to $\tau^2$ or smaller.
Such probabilities become negligible as $\tau$ approaches zero, and thus can be ignored in the limit.
Rearranging and letting $\tau \to 0$, we obtain the following recursive expression for $V_{ik}(\varsigma_{i})$:
\begin{multline}
  \label{eq:bellman}
  V_{ik}(\varsigma_{i})
  = \frac{1}{\rho_i + \sum_{l \neq k} q_{kl} + \sum_m \lambda_{mk}}
  \times \left[ u_{ik} +
    \sum_{l \neq k} q_{kl} V_{il}(\varsigma_{i}) + \right. \\ \left.
    \sum_{m \neq i} \lambda_{mk} \sum_{j=0}^{J-1} \varsigma_{imjk} V_{i,l(m,j,k)}(\varsigma_{i}) +
    \lambda_{ik} \E \max_j \{ \psi_{ijk} + \varepsilon_{ijk} + V_{i,l(i,j,k)} (\varsigma_{i})\}
  \right]
\end{multline}
The denominator contains the sum of the discount factor and the rates
of all events that might possibly change the state.
The numerator is composed of the flow payoff for being in state $k$,
the rate-weighted values associated with exogenous state changes,
the rate-weighted values associated with states that occur after moves by rival players,
and the expected current and future value obtained when a move arrival for player $i$
occurs in state $k$.
The expectation is taken with respect to the joint distribution of
$\varepsilon_{ik} = ( \varepsilon_{i0k}, \dots, \varepsilon_{i,J-1,k} )^\top$.

\begin{remark}
  Note that the $\E\max$ term in \eqref{eq:bellman} can be written in
  the usual ``log-sum-exp'' form when the errors satisfy
  Assumption~\ref{assp:tiev}:
  \begin{equation*}
    \E \max_j \{ \psi_{ijk} + \varepsilon_{ijk} + V_{i,l(i,j,k)}(\varsigma_i) \}
    = \ln \sum_j \exp\left(\psi_{ijk} + V_{i,l(i,j,k)}(\varsigma_i)\right).
  \end{equation*}
\end{remark}

\begin{ex:renewal}[continued]
In the renewal model, the value function can be expressed very simply as follows (where the $i$ subscript and beliefs have been omitted since $\Nplayers = 1$):
\begin{equation*}
  V_k = \frac{1}{\rho + \gamma + \lambda} \left( u_k + \gamma V_{k+1}
  + \lambda_k \E \max \left\{ \varepsilon_{0k} + V_k, \mu + \varepsilon_{1k} + V_1 \right\} \right).
\end{equation*}
\end{ex:renewal}

\begin{ex:entry}[continued]
  In the $2 \times 2$ entry model, the value function for
  player $1$ in state $k$, where $x_k = (x_{k1}, x_{k2}, d_{k}) \in \{ 0, 1 \} \times \{ 0, 1 \} \times \{ \text{L}, \text{H} \}$,
  can be expressed recursively as (omitting beliefs $\varsigma_1$ for brevity):
  \begin{multline*}
    V_{1k}
    = \frac{1}{\rho_1 + \1\{d_{k} = \text{L}\} \gamma_{\text{LH}} + \1\{d_{k} = \text{H}\} \gamma_{\text{HL}} + \lambda_{1k} + \lambda_{2k}} \\
    \times \left(
      u_{1k}
      + \1\{d_{k} = \text{L}\} \gamma_{\text{LH}} V_{1,l(0, \text{H}, k)}
      + \1\{d_{k} = \text{H}\} \gamma_{\text{HL}} V_{1,l(0, \text{L}, k)}
      + \lambda_{2k} \varsigma_{120k} V_{1k}
      \right.\\\left.
      + \lambda_{2k} \varsigma_{121k} V_{1,l(2, 1, k)}
      + \lambda_{1k} \E \max\left\{
        \varepsilon_{i0k} + V_{1k},
        \psi_{11k} + \varepsilon_{11k} + V_{1,l(1, 1, k)}
      \right\}
    \right),
  \end{multline*}
  where $l(0,\text{H},k)$ and $l(0,\text{L},k)$ are the continuation
  states when nature switches the level of demand to $\text{H}$ and
  $\text{L}$, respectively, when in state $k$.
  $\varsigma_{12jk}$ is firm 1's belief about firm 2 choosing $j$.
\end{ex:entry}

\subsection{Markov Perfect Equilibrium}

\begin{definition}
A \emph{Markov perfect equilibrium} is a collection of stationary Markov policy rules $\{ \delta_i^* \}_{i=1}^\Nplayers$ such that for each player $i$ and for all $(k, \varepsilon_{ik})$, $\delta_i^*(k, \varepsilon_{ik}) = b_i(k, \varepsilon_{ik}, \varsigma_i)$ and $\varsigma_{imjk} = \Pr\left[ \delta_m^*(k, \varepsilon_{mk}) = j \mid k \right]$ for all $m \neq i$.
\end{definition}

Following the literature, we focus on Markov perfect equilibria.
The definition requires that for each player $i$, $\delta_i^*$ is a best response in all states given the beliefs $\varsigma_i$ and that these beliefs are consistent with the strategies $\delta_m^*$ for each rival player $m$.

Following \cite{milgrom-weber-1985} and \cite{aguirregabiria-mira-2007},
we characterize Markov perfect equilibria in terms of
equilibrium CCPs
\begin{equation}
  \sigma_{ijk} = \Pr\left[ \delta_i^*(k, \varepsilon_{ik}) = j \mid k \right].
\end{equation}
Henceforth, we will denote equilibrium choice probabilities and corresponding beliefs by $\sigma_{ijk}$.
Thus, $\sigma = (\sigma_1, \dots, \sigma_\Nplayers)$ will denote a profile of equilibrium choice probabilities and $\sigma_{-i}$ will denote the collection of rival choice probabilities that constitute player $i$'s beliefs.

ABBE proved that such an equilibrium exists when players share common move arrival and discount rates and when the move arrival rates do not vary across states (i.e., $\lambda_{ik} = \lambda$ and $\rho_i = \rho$ for all $i$ and $k$).
The following theorem extends this to the present generalized model with heterogeneity.
The proof, and all others, appears in \autoref{sec:proofs}.

\begin{theorem}
  \label{thm:existence}
  If Assumptions~\ref{assp:dx}--\ref{assp:pi} hold,
  then a Markov perfect equilibrium exists.
\end{theorem}

\subsection{Linear Representation of the Value Function}

Before proceeding, we revisit one of the central results of
ABBE (Proposition 2), a continuous-time analog of
\citet[][Proposition 1]{hotz93conditional} for discrete-time models.
Restated below as Lemma~\ref{lem:ccp:inversion}, ABBE showed that
differences in choice-specific value functions---that is,
$[\psi_{ijk} + V_{i,l(i,j,k)}(\sigma)] - [\psi_{ij'k} + V_{i,l(i,j',k)}(\sigma)]$
for two choices $j$ and $j'$---are identified directly
as functions of the CCPs $\sigma_i$.

\begin{lemma}[ABBE, 2016, Proposition 2]
  \label{lem:ccp:inversion}
  Under Assumptions~\ref{assp:dx}--\ref{assp:pi}, for each player $i$,
  state $k$, and choice $j$
  the choice-specific value function is identified, up to differences with respect
  to some baseline choice $j'$, as a function of the
  CCPs. Specifically, there exists a known function $\Phi$ such that:
  \begin{equation}
    \label{eq:csvfdiff}
    \psi_{ijk} + V_{i,l(i,j,k)}(\sigma) = \psi_{ij'k} + V_{i,l(i,j',k)}(\sigma) + \Phi(j,j',\sigma_{ik}).
  \end{equation}
\end{lemma}

This result will prove useful for vectorizing the value function.
Let $\Sigma_{m}(\sigma_{m})$ denote the transition matrix implied by
the CCPs $\sigma_{m}$ and the continuation
state function $l(m,\cdot,\cdot)$.
That is, the $(k,l)$ element of the matrix $\Sigma_{m}(\sigma_{m})$ is the
probability of transitioning from state $k$ to state $l$ as a result of
an action by player $m$ given the CCPs $\sigma_{m}$.
Let $Q_0 = (q_{kl})$ denote the intensity matrix for exogenous state
transitions and let
$\tilde{Q}_0 = Q_0 - \diag(q_{11}, \dots, q_{KK})$
be the matrix formed by taking $Q_0$ and setting the diagonal elements
to zero.

With this notation and Lemma~\ref{lem:ccp:inversion} in hand,
following \eqref{eq:bellman}, for given equilibrium CCPs
$\sigma$ we define the operator $\Gamma_i^{\sigma}$ as
\begin{equation}
  \label{eq:bellman:matrix}
  \Gamma_i^{\sigma}(V_i) = D_i \left[ u_i + \tilde{Q}_0 V_i + \sum_{m \neq i} L_m \Sigma_m(\sigma_m) V_i
  + L_i \left\{ \Sigma_i(\sigma_i) V_i + C_i(\sigma_i) \right\}\right],
\end{equation}
where $D_i$ is the $K \times K$ diagonal matrix with elements
$(D_i)_{kk} = 1/(\rho_i + \sum_{l \neq k} q_{kl} + \sum_{m=1}^\Nplayers \lambda_{mk})$,
$L_m = \diag(\lambda_{m1}, \dots, \lambda_{mK})$ is a diagonal matrix
containing the move arrival rates for player $m$,
$C_i(\sigma_{i})$ is the $K \times 1$ vector containing
the ex-ante expected value of the instantaneous payoff
$c_{ijk} = \psi_{ijk} + \varepsilon_{ijk}$ for player $i$ in each state $k$
given the best response probabilities $\sigma_i$.
That is, $k$-th element of $C_i(\sigma_i)$ is
$\sum_{j=0}^{J-1} \sigma_{ijk} \left[ \psi_{ijk} + e_{ijk}(\sigma_i) \right]$,
where $e_{ijk}(\sigma_i)$ is the expected value of $\varepsilon_{ijk}$
given that action $j$ is chosen:
\begin{equation*}
  e_{ijk}(\sigma_i) \equiv
  \frac{1}{\sigma_{ijk}} \int \varepsilon_{ijk}
  \cdot \1\left\{ \varepsilon_{ij'k} - \varepsilon_{ijk}
    \leq \psi_{ijk} - \psi_{ij'k} + V_{i,l(i,j,k)}(\sigma) - V_{i,l(i,j',k)}(\sigma) \;
    \forall j' \right\}\,
  dF_{ik}(\varepsilon_{ik}).
\end{equation*}
By Lemma~\ref{lem:ccp:inversion}, the choice-specific value differences on the right-hand
side are in turn functions of player $i$'s CCPs $\sigma_i$.
Hence, holding fixed the equilibrium beliefs $\sigma$, the corresponding
value function is a fixed point of
$\Gamma_i^{\sigma}$: $V_i = \Gamma_i^{\sigma}(V_i)$.

\begin{remark}
  Although $e_{ijk}(\sigma_i)$ involves a multivariate integral,
  \cite{aguirregabiria-mira-2002,aguirregabiria-mira-2007}
  established closed forms in terms of choice probabilities in two leading cases.
  For the case of Assumption~\ref{assp:tiev}, we have
  $e_{ijk}(\sigma_i) = \gamma_{\text{EM}} - \ln \sigma_{ijk}$,
  where $\gamma_{\text{EM}}$ is the Euler-Mascheroni constant ($\gamma_{\text{EM}} \approx 0.5772$).
  For $J = 2$ choices and $\varepsilon_{ik} \sim \Normal(0,\Omega)$,\footnote{See
  \cite{aguirregabiria-mira-2007} equation 13 and footnote 7 for details.}
  \begin{equation*}
    e_{ijk}(\sigma_i) = \frac{
        \var(\varepsilon_{ijk}) - \cov(\varepsilon_{i0k}, \varepsilon_{i1k})
      }{
        \sqrt{\var(\varepsilon_{i1k} - \varepsilon_{i0k})}
      }
      \frac{\phi\left(\Phi^{-1}(\sigma_{ijk})\right)}{\sigma_{ijk}},
  \end{equation*}
  where $\Phi$ and $\phi$ denote, respectively, the standard normal cdf and pdf.
\end{remark}


Collecting terms involving $V_i$ in \eqref{eq:bellman:matrix} and solving
leads to a linear representation of the value function in terms of
CCPs, rate parameters, and payoffs as
formalized in the following Theorem.
This representation generalizes Proposition 6 of ABBE and forms the
basis of the identification result later in Section~\ref{sec:identification:u}.
It is analogous to a similar result for discrete time games by
\citet[][eq. 6]{pesendorfer08asymptotic}.

\begin{theorem}
  \label{thm:v:linear}
  If Assumptions~\ref{assp:dx}--\ref{assp:pi} hold, then for a given
  collection of equilibrium choice probabilities $\sigma$, $V_i$ has the
  following linear representation for each $i$:
  \begin{align}
    V_i(\sigma) &= \Xi_i^{-1}(\sigma) \left[ u_i + L_i C_i(\sigma_i) \right] \qquad \text{where}     \label{eq:v:linear} \\
    \Xi_i(\sigma) &= \rho_i I_K + \sum_{m=1}^\Nplayers L_m [I_K - \Sigma_m(\sigma_m)] - Q_0    \label{eq:Xi_i}
  \end{align}
  is a nonsingular $K \times K$ matrix and $I_K$ is the $K \times K$ identity matrix.
\end{theorem}

\proofapx{}

\subsection{Continuous Time Markov Jump Processes Representation}

The model's reduced form is a finite state Markov jump process, a stochastic process $X(t)$ indexed by $t \in [0, \infty)$ taking values in a finite state space $\mathcal{X} = \lbrace 1, \dots, K \rbrace$.
If we observe this process at time $t$ and state $X(t)$, it remains in this state for a random duration $\tau$ before transitioning to another state $X(t + \tau)$.
The duration $\tau$ is the holding time.
A trajectory of this process is a piecewise-constant, right-continuous function of time.
Jumps occur according to a Poisson process and holding times between jumps are exponentially distributed.
For fundamental properties of Markov jump processes, see \citet[Section 4.8]{karlin75first} or \citet[part II]{chung-1967}.

A finite Markov jump process can be summarized by its \emph{intensity matrix} or \emph{infinitesimal generator matrix}.
Consider the intensity matrix for nature, $Q_0 = (q_{kl})$
where for $k \neq l$
$q_{kl} = \lim_{h \to 0} \frac{\Pr\left[X(t+h) = l \mid X(t) = k \right]}{h}$
is the probability per unit of time that the system transitions from state $k$ to $l$ and the diagonal elements are $q_{kk} = -\sum_{l \neq k} q_{kl}$ so that the row sums equal zero.
Holding times before leaving state $k$ follow an exponential distribution with rate parameter $-q_{kk}$.
Conditional on leaving state $k$, the system transitions to state $l \neq k$ with probability
$q_{kl} / \sum_{l \neq k} q_{kl} = -q_{kl} / q_{kk}$.

For discrete time data, the exact times when actions and
state changes occur are unobserved.
With equispaced data (e.g., annual or quarterly) only the states at
the beginning and end of each period of length $\Delta$ are observed.
Although we cannot know the exact sequence of actions and state changes,
the model allows us to determine the likelihood of any transition
occurring over a time interval of length $\Delta$ using the
\emph{transition matrix}, denoted $P(\Delta)$.

Let $P_{kl}(\Delta) = \Pr\left[X(t+\Delta) = l \mid X(t) = k\right]$ denote the probability that the system is in state $l$ after a period of length $\Delta$ given that it was initially in state $k$.
The transition matrix $P(\Delta) = (P_{kl}(\Delta))$ is the corresponding $K \times K$ matrix of these probabilities and is given by the matrix exponential of the intensity matrix $Q$ scaled by the time interval $\Delta$:
\begin{equation}
  \label{eq:interval_transition_probabilities}
  P(\Delta) = \exp(\Delta Q) = \sum_{j=0}^\infty \frac{(\Delta Q)^j}{j!}.
\end{equation}
This is the matrix analog of the scalar exponential $\exp(x)$ for $x \in \R$.\footnote{Although we cannot calculate the infinite sum \eqref{eq:interval_transition_probabilities} exactly, we can compute $\exp(\Delta Q)$ numerically using known algorithms implemented in the Fortran package Expokit \citep{sidje98expokit} or the \texttt{expm} command in Matlab. See \cite{sherlock21direct} for recent discussion of the uniformization method.}

In the dynamic games we consider, the state space dynamics can be fully characterized by $\Nplayers + 1$ competing Markov jump processes with intensity matrices $Q_0, Q_1, \dots, Q_\Nplayers$ for nature and each player.
The \emph{aggregate intensity matrix} is defined as $Q \equiv Q_0 + Q_1 + \dots + Q_\Nplayers$.

\begin{ex:renewal}[continued]
Consider the $Q$ matrix implied by the continuous-time single-agent renewal model.
The state variable in the model is the total accumulated mileage of a bus engine, $\mathcal{K} = \{ 1, \dots, K \}$.
The exogenous state transition process is characterized by a $K \times K$ intensity matrix $Q_0$ on $\mathcal{K}$ with one parameter, $\gamma$, governing the rate of mileage increases:
\begin{equation*}
  Q_0= \begin{bmatrix}
    -\gamma & \gamma & 0 & 0 & \cdots & 0 \\
    0 & -\gamma & \gamma & 0 & \cdots & 0 \\
    \vdots & \vdots & \ddots & \vdots & \vdots & \vdots \\
    0 & 0 & \ldots & -\gamma & \gamma & 0 \\
    0 & 0 & \ldots & 0 & -\gamma & \gamma \\
    0 & 0 & \ldots & 0 & 0 & 0
  \end{bmatrix}.
\end{equation*}
The intensity matrix for state changes induced by the agent is
\begin{equation*}
  Q_1 = \begin{bmatrix}
    0 & 0 & 0 & \cdots & 0 & 0 \\
    \lambda_{\text{L}} \sigma_{12} & -\lambda_{\text{L}} \sigma_{12} & 0 & \cdots & 0 & 0 \\
    \lambda_{\text{L}} \sigma_{13} & 0 & -\lambda_{\text{L}} \sigma_{13} & \cdots & 0 & 0 \\
    \vdots & \vdots & \vdots & \ddots & \vdots & \vdots \\
    \lambda_{\text{H}} \sigma_{1,K-1} & 0 & 0 & \cdots & -\lambda_{\text{H}} \sigma_{1,K-1} & 0 \\
    \lambda_{\text{H}} \sigma_{1K} & 0 & 0 & \cdots & 0 & -\lambda_{\text{H}} \sigma_{1K} \\
    \end{bmatrix}.
\end{equation*}
The aggregate intensity matrix in this case is $Q = Q_0 + Q_1$:
\begin{equation}
  \label{eq:rust:Q}
 Q = \begin{bmatrix}
    -\gamma & \gamma & 0 & \cdots & 0 & 0 \\
    \lambda_{\text{L}} \sigma_{12} & -\lambda_{\text{L}} \sigma_{12}-\gamma & \gamma & \cdots & 0 & 0 \\
    \lambda_{\text{L}} \sigma_{13} & 0 & -\lambda_{\text{L}} \sigma_{13}-\gamma & \cdots & 0 & 0 \\
    \vdots & \vdots & \vdots & \ddots & \vdots & \vdots \\
    \lambda_{\text{H}} \sigma_{1,K-1} & 0 & 0 & \cdots & -\lambda_{\text{H}} \sigma_{1,K-1} - \gamma & \gamma \\
    \lambda_{\text{H}} \sigma_{1K} & 0 & 0 & \cdots & 0 & -\lambda_{\text{H}} \sigma_{1K} \\
    \end{bmatrix}.
\end{equation}
\end{ex:renewal}

\begin{ex:entry}[continued]
Let $h_{ik}$ be the hazard of player $i$ switching from active to inactive or vice versa in state $k$.
We have dropped the $j$ subscript here for notational simplicity since $j=0$ does not change the state.
Let $\gamma_{\text{LH}}$ and $\gamma_{\text{HL}}$ be the rates at which nature switches between demand states (i.e., demand moves from low to high at rate $\gamma_{\text{LH}}$).
The aggregate state space dynamics are illustrated in \autoref{fig:2x2x2-entry}.
Recall that the reduced form hazards $h_{ik}$ of firm $i$ taking action $j=1$ in state $k$ are related to the structural quantities through the relation $h_{ik} = \lambda_{ik} \sigma_{ijk}$.

\begin{figure}[tb]
 \centering
 \resizebox{0.8\textwidth}{!}{\includegraphics{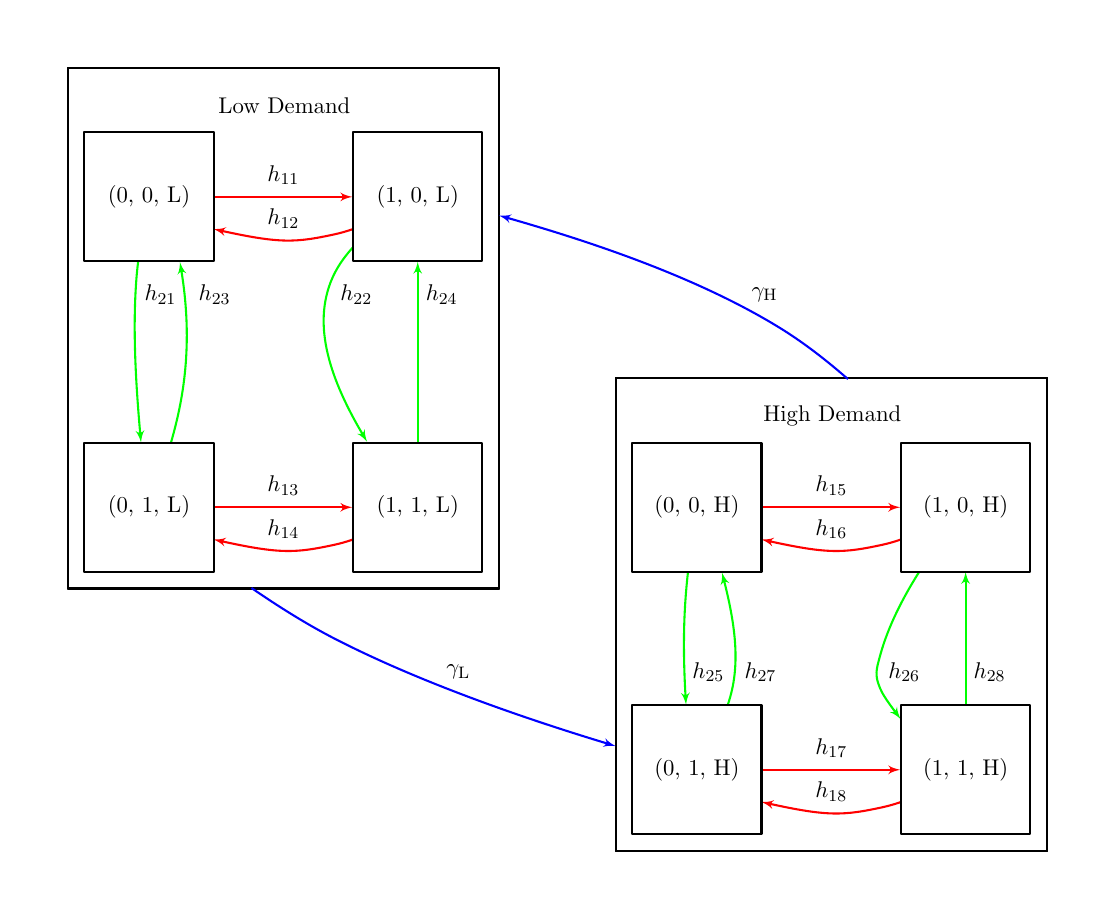}}\\
 \medskip
 \begin{footnotesize}
  Two demand states, $\text{L}$ and $\text{H}$, two firms, and two choices ($j = 0$ continue, $j=1$ switch).
  Reduced form hazards $h_{ik}$, denoting the rates of switching ($j=1$),
  are related to the structural quantities as $h_{ik} = \lambda_{ik} \sigma_{i1k}$.
 \end{footnotesize}
 \caption{Two Player Entry Game with Exogenous Demand State}
 \label{fig:2x2x2-entry}
\end{figure}

The state transition hazards can be characterized by an $8 \times 8$ intensity matrix $Q$.
Note that firms cannot change the demand state, firms cannot change each other's states,
and nature cannot change the firms' states.
Therefore, the overall intensity matrix is a sparse matrix of the form
\begin{small}
\begin{equation}
\label{eq:2x2x2:Q}
Q = \left[
\begin{array}{cccc|cccc}
  \cdot & h_{11} & h_{21} & 0 & \gamma_{\text{L}} & 0 & 0 & 0 \\
  h_{12} & \cdot & 0 & h_{22} & 0 & \gamma_{\text{L}} & 0 & 0 \\
  h_{23} & 0 & \cdot & h_{13} & 0 & 0 & \gamma_{\text{L}} & 0 \\
  0 & h_{24} & h_{14} & \cdot & 0 & 0 & 0 & \gamma_{\text{L}} \\
  \hline
  \gamma_{\text{H}} & 0 & 0 & 0 & \cdot & h_{15} & h_{25} & 0 \\
  0 & \gamma_{\text{H}} & 0 & 0 & h_{16} & \cdot & 0 & h_{26} \\
  0 & 0 & \gamma_{\text{H}} & 0 & h_{27} & 0 & \cdot & h_{17} \\
  0 & 0 & 0 & \gamma_{\text{H}} & 0 & h_{28} & h_{18} & \cdot \\
\end{array}
\right],
\end{equation}
\end{small}
The diagonal elements have been omitted for brevity.
Importantly, the locations of the nonzero elements are distinct because the state-to-state communication patterns differ.
Therefore, given $Q$ we can immediately determine $Q_0$, $Q_1$, and $Q_2$.
\end{ex:entry}

\subsection{Comparison with Discrete Time Models}
\label{sec:ctdt}

We conclude with remarks on continuous time and discrete time models.
First, consider a typical discrete time model where agents move simultaneously and the period between decisions is calibrated to the sampling period of the data.
In an entry/exit setting where the choice set is $\mathcal{J} = \{ 0, 1 \}$, there is exactly one entry or exit decision per year.
In discrete time data, passive continuation actions are coded as active decisions, but in reality they represent the absence of an active choice during the period.
Consider a chain store setting where the choice is the net number of stores to open during the year, $\mathcal{J} = \{ -J, \dots, J \}$.
This implies at most $J$ openings or closings per year.
Hence, $J$ must be chosen to be the maximum number of possible stores opened or closed by any chain firm in any period.

Now consider a continuous time model with a common move arrival rate $\lambda$ for all players and states.
In the entry/exit setting, the choice set is still $\mathcal{J} = \{ 0, 1 \}$ which implies \emph{on average} $1/\lambda$ entries or exits per year.
Multiple entries and exits are allowed, and the parameters imply a distribution over the number and type of such events.
The choice set represents possible \emph{instantaneous} state changes, so in the chain store expansion example, if no more than one store opens or closes simultaneously, we specify $\mathcal{J} = \{ -1, 0, 1 \}$.
This implies \emph{on average} at most $1 / \lambda$ openings or closings per year.
In our continuous time model the rate $\lambda$ is a free parameter that can adjust to match the data, not imposing an ad hoc restriction on the number of actions per unit of time.
The time-aggregated implications of the continuous time model are not functionally different if we change the time period and are unrelated to the sampling period, a feature of the data collection process.

\section{Identification}
\label{sec:identification}

Due to time aggregation, our identification analysis
separates the data issue---that we may only observe $P(\Delta)$
instead of $Q$---from recovering
the structural parameter $\theta$ from the continuous-time
reduced form $Q$.
We proceed in two steps; researchers with continuous
time data can begin with the second step.\footnote{Although we
  take a sequential approach to identification,
  a direct approach from $P(\Delta)$ to $\theta$ may
  also be possible.
  While a sequential approach could, in principle,
  impose additional constraints when identifying $Q$ from $P(\Delta)$ that are unnecessary
  for directly identifying $\theta$ from $P(\Delta)$, in this case we exploit
  non-parametric structural restrictions from the model to identify $Q$ in the first step,
  mitigating this concern.}
Deriving the structural model's implications is a
bottom-up exercise: the structural primitives $u$ and $\psi$ imply
value functions $V$ which imply choice probabilities $\sigma$.
These probabilities with move rates $\lambda$ and
state transitions by nature $Q_0$ imply an intensity matrix
$Q$.
Given the $Q$ matrix and a sampling process, this
implies a data generating process.
For a fixed sampling interval $\Delta$ the distribution
of observable data is $P(\Delta) = \exp(\Delta Q)$.

The identification problem requires us to consider
the inverse problem.
These steps are represented in \autoref{fig:ident}.
If the complete continuous time process is observable, then
$Q$ is trivially identified and we can move to identification of
the structural model.
However, for discrete time data we must use our
knowledge of the data generating process, represented
by the transition matrix $P(\Delta)$ for an interval $\Delta$, to
derive conditions under which we can uniquely determine the reduced
form intensity matrix $Q$.
We show this is possible under mild conditions using
restrictions that the structural model places on the $Q$ matrix.

\begin{figure}[tb]
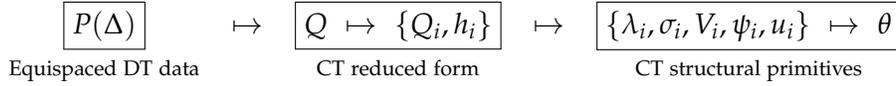

  \begin{tabular}{ccccc}
    \fbox{$P(\Delta)$} & $\mapsto$ & \fbox{$Q\,\,\mapsto\,\,\{ Q_i, h_i \}$} & $\mapsto$ & \fbox{$\{ \lambda_i, \sigma_i, V_i, \psi_i, u_i \}\,\,\mapsto\,\,\theta$} \\
    \scriptsize{Equispaced DT data} & & \scriptsize{CT reduced form} & & \scriptsize{CT structural primitives} \\
  \end{tabular}
  \caption{Identification Analysis}
  \label{fig:ident}
\end{figure}

With $Q$ in hand, we turn to the structural
primitives of the model:
the flow payoffs $u$ and instantaneous payoffs $\psi$.
We show that knowledge of $Q$ allows us to recover these structural primitives with a smaller number of additional identifying restrictions than are required in discrete time models.
This is due to the absence of simultaneous moves at any given instant, which is also the source of the computational efficiency of the model.

\subsection{Identification of $Q$}
\label{sec:identification:Q}

With continuous-time data, identification and estimation of the
intensity matrix for finite-state Markov jump processes
is well-established
\citep{billingsley61statistical}.
However, when a continuous-time process is only sampled at discrete
points in time, the parameters of the underlying continuous-time model
may not be point identified.\footnote{This is known as the
  \emph{aliasing problem} and has been studied in the
  context of continuous-time systems of stochastic differential equations
  \citep{sims-1971, phillips-1973, hansen-sargent-1983,
    hansen-sargent-1991-ch9, geweke-1978, kessler-rahbek-2004,
    mccrorie-2003, blevins-2017}.
  See Figure 1 of \cite{blevins-2017} for an illustration in the
  frequency domain, where the problem is perhaps most obvious.}
In the present model, the concern is that there may be multiple $Q$
matrices which give rise to the same data generating process,
which is the transition probability matrix $P(\Delta)$ in the
case of fixed sampling at an interval $\Delta$.%
\footnote{A related issue is the embeddability problem: could the
  transition matrix $P(\Delta)$ have been generated by a
  continuous-time Markov jump process for some intensity matrix $Q$ or
  some discrete-time chain over fixed time periods of length $\delta$?
  This problem was first proposed by \cite{elfving37zur}.
  \cite{kingman62imbedding} derived the set of embeddable processes
  with $K = 2$ and \cite{johansen74some} gave an explicit description
  of the set for $K = 3$.
  \cite{singer-spilerman-1976} summarize several known necessary
  conditions for embeddability involving testable conditions on the
  determinant and eigenvalues of $P(\Delta)$.
  We assume throughout that the continuous time model is
  well-specified and that such an intensity matrix exists.}

In discrete time settings, a similar identification problem is
masked when assuming the unknown frequency of moves equals
the (known) sampling frequency \citep{hong15aggregated}.
Suppose agents move at intervals of length $\delta$ with transition
matrix $P_0$ while the data sampling interval is $\Delta > \delta$.
Then the mapping between the data (equispaced observations at length $\Delta$)
and the transition matrix is: $P(\Delta) = P_0^{\Delta/\delta}$.
Generally, there are multiple solutions to this equation
\citep{gantmacher-1959, singer-spilerman-1976}, meaning
identification of $P_0$ is non-trivial.

Previous work on this identification problem seeks conditions on
the observable discrete-time transition matrix $P(\Delta)$.
We briefly review these results in the next subsection, but
our approach is to show that one can identify $Q$ via
identifying restrictions on the primitives of the underlying
structural model and that such restrictions arise from the model itself.
These can be viewed as exclusion restrictions.

For example, in applications there are typically player-specific
components of the state vector where player $i$ cannot
change the player-specific state of player $j$ and vice-versa.
In an entry-exit model, such a state is incumbency status: players can
enter and exit by their own action, but no player can enter or exit on
behalf of another.
Similarly, if the state vector has components that are
exogenous state variables, such as population, then any
state changes involving those variables must be due to
nature and not by any player.
This natural structure implies many linear restrictions on the $Q$ matrix.
We show that restrictions of this form serve to limit the domain of the mapping
$Q \mapsto \exp(\Delta Q) = P(\Delta)$
to guarantee the intensity matrix $Q$
is identified.

\subsubsection{Identification of Unrestricted $Q$ Matrices}

Returning to identification of $Q$, recall that
the question is whether there exists a unique matrix $Q$ that leads to
the observed transition matrix $P(\Delta) = \exp(\Delta Q)$ when observations are
sampled at intervals of length $\Delta$.
The matrix logarithm $\ln P(\Delta)$ is not unique in general
\citep[see][]{gantmacher-1959, singer-spilerman-1976},
so the question amounts to finding conditions for a unique solution.

Previous, mathematical treatments view the relationship
$\exp(\Delta Q) = P(\Delta)$ from the perspective of the transition matrix $P(\Delta)$.
In such cases there is no underlying model that generates $Q$,
so $Q$ itself is the primitive of interest and is unrestricted
(subject to being a valid intensity matrix).
Most previous work focused on finding sufficient conditions on
the matrix $P(\Delta)$ to guarantee that $\ln P(\Delta)$ is unique.
For example, if the eigenvalues of $P(\Delta)$ are distinct, real, and
positive, then $Q$ is identified \citep{culver-1966}.
More generally, \cite{culver-1966} proved that $Q$ is identified if
the eigenvalues of $P(\Delta)$ are positive and no elementary divisor
(Jordan block) of $P(\Delta)$ appears more than once.
Other sufficient conditions include
$\min_k\{ P_{kk}(\Delta) \} > 1/2$ \citep{cuthbert-1972}
and
$\det P(\Delta) > \e^{-\pi}$ \citep{cuthbert-1973}.
See \cite{singer-spilerman-1976} for a summary.

Other conditions involve alternative sampling schemes.
For example, $Q$ is identified if the sampling
$\Delta$ is sufficiently small \citep{cuthbert-1973,
  singer-spilerman-1976, hansen-sargent-1983}.
Alternatively, $Q$ is identified if
the process is sampled at distinct intervals $\Delta_1$ and
$\Delta_2$ where $\Delta_2 \neq k \Delta_1$ for any integer $k$
\citep[][5.1]{singer-spilerman-1976}.

The first conditions---restrictions on $P(\Delta)$---are
based on a ``top down'' approach and are undesirable when
$Q$ is generated by an underlying model.
The second conditions are based on how the
continuous time process is sampled, which cannot be changed
if data are already collected at regular intervals.
Instead, we take a ``bottom up'' approach which allows economic theory
to inform our identification conditions via restrictions on $Q$ that
guarantee uniqueness of $\ln P(\Delta)$.
More compelling conditions involve cross-row and cross-column restrictions on the $Q$ matrix and
known zeros of the $Q$ matrix.
Such restrictions arise naturally once players, actions, and resulting state transitions are
defined.

\subsubsection{Structural Restrictions for Identification of $Q$}

The problem of identifying continuous time models with only discrete
time data has appeared previously in the econometrics literature,
in work by \cite{phillips-1973} on continuous time regression models.
He considered multivariate, continuous-time, time-homogeneous
regression models of the form
$y'(t) = A y(t) + \xi(t)$,
where $y(t)$ is an $n \times 1$ vector and $A$ is an $n \times n$
structural matrix.
He discussed the role of prior information on $A$ and how
it can lead to identification.
He showed that $A$ is identified given discrete time observations
on $y$ if $A$ satisfies certain rank conditions.

Our identification strategy is inspired by this work, but our model
differs because the $Q$ matrix is an intensity matrix
(rather than an arbitrary matrix of regression coefficients) and has
sparse structure dictated by an underlying structural
model.
Yet, there are similarities: the present model can be
characterized by a system of differential equations, where the intensity
matrix $Q$ plays a role similar to the matrix $A$.
If $Q$ is a valid intensity matrix, then the functions $P(\Delta)$
solving this system are the transition matrices of continuous-time
stationary Markov chains \cite[][p. 251--257]{chung-1967}.

The structural model restricts $Q$ to a lower-dimensional subspace
since it is sparse and must satisfy both within-row and across-row
restrictions, and given the results above it seems likely that these
restrictions could lead to identification of $Q$.
That is, even if there are multiple matrix solutions to the equation
$P(\Delta) = \exp(\Delta Q)$, it is unlikely that two of them
simultaneously satisfy the restrictions of the structural model.
We return to the two examples introduced previously to illustrate this idea.

\begin{ex:renewal}[continued]
In the single-agent renewal model the aggregate intensity matrix is given in
\eqref{eq:rust:Q}.
The number of nonzero hazards is substantially less than the total.
Consider $K = 90$: there are $90^2 - 90 = 8,010$ non-trivial state-to-state transitions.
Only 178 are permitted at any instant: 89 due to nature and 89 by player action.
The remaining 7,832 transitions are not possible in a single event.
Nature cannot decrease mileage and can only increase it by one state
at a given instant (although multiple state jumps are possible over an
interval).
The agent can only reset mileage to the initial state.
Therefore, there are 7,832 known zeros of $Q$---more than sufficient to uniquely identify $Q$.
Given $Q$, we can separately determine both $Q_0$ and $Q_1$.
The choice-specific hazards $h_{1k}$ are products of the
move arrival rates and conditional choice probabilities, which
introduces shape restrictions on $h_{1k} = \lambda_k\sigma_{1k}$ across states $k$.
\end{ex:renewal}

\begin{ex:entry}[continued]
In the $2 \times 2 \times 2$ entry example, the aggregate intensity matrix is
given by \eqref{eq:2x2x2:Q}.
Some transitions cannot happen at all, such as $(0,1,\text{L})$ to $(1,0,\text{L})$.
The remaining transitions can happen only due to the action of one of the firms, but not the other.
For example, moving from $(0,0,\text{H})$ to $(1,0,\text{H})$ is only possible if firm $1$ chooses to become active.
From any state, the set of other states to which either firm can move the state as a result of an action is limited naturally by the model and the definition of the state space.
This structure yields intensity matrices that are sparse, which makes identification of $Q$ more likely even with time aggregation since any observationally equivalent $Q$ matrix must have the same sparsity pattern.
Finally, given $Q$ we can again separately recover $Q_0$, $Q_1$, and $Q_2$.
\end{ex:entry}

Similar sparse structures arise in even models with large numbers of players and
millions of states, as in the application of ABBE.
In light of this lower-dimensional structure, we build on the results of
\cite{blevins-2017} who gave sufficient conditions for identification in
first-order linear systems of stochastic differential equations.
We apply those results to the case of finite-state Markov jump processes
generated by our structural model.

The key insight is that structural restrictions on $Q$
of the form $R \vectorize(Q) = r$ can rule out alternative solutions to the
matrix exponential equation $\exp(\Delta \tilde{Q}) = P(\Delta)$.
When there are sufficiently many linearly independent restrictions, the
intensity matrix $Q$ is uniquely identified.
Adapting Theorem 1 of \cite{blevins-2017} to finite-state Markov jump
processes, we require at least $\floor{\frac{K-1}{2}}$ linear restrictions
when $R$ has full rank.

The following theorem establishes that there are sufficiently many full rank
restrictions to identify $Q$ in a broad class of games.
This theorem includes exogenous market-specific state variables and shows
that such states increase the number of zero restrictions and make
identification of $Q$ more likely, as do player-specific state
variables.

\begin{theorem}[Identification of $Q$]
  \label{thm:id_Q}
  Suppose the state vector is $x = (x_0, x_1, \dots, x_\Nplayers) \in \mathcal{X}_0 \times \mathcal{X}_1 \times \dots \times \mathcal{X}_\Nplayers$ where the component $x_0 \in \mathcal{X}_0$ is an exogenous market characteristic taking $\abs{\mathcal{X}_0} = K_0$ values and for each $i = 1, \dots, \Nplayers$ the component $x_i$ is a player-specific state affected only by the action of each player with $\abs{\mathcal{X}_i} = K_1$ possible distinct values.
  If $Q$ has distinct eigenvalues that do not differ by an integer multiple of $2 \pi i/\Delta$, then $Q$ is identified when
  \begin{equation}
    \label{eq:id_Q}
    K_0 K_1^\Nplayers - K_0 - \Nplayers J + \frac{1}{2} \geq 0.
  \end{equation}
  The quantity on the left is strictly increasing in $K_1$,
  strictly increasing in $K_0$ when $K_1 > 1$, and
  strictly decreasing in $J$.
\end{theorem}

The sparsity of $Q$ helps and is increasing in both the number of exogenous states $K_0$ and player-specific states $K_1$, but decreasing in the number of choices $J$.
Therefore, for identification we need either a sufficiently large number of states or a sufficiently small number of choices.
Fortunately, in most applications $J$ is small relative to $K$---particularly in continuous time models as discussed in Section~\ref{sec:ctdt}.

\begin{ex:entry}[continued]
Our running entry model example is a binary choice game
with $\Nplayers = 2$, $J = 2$, $K_0 = 2$, and $K_1 = 2$,
so by \autoref{thm:id_Q} $Q$ is identified.
\end{ex:entry}

Furthermore, we can see that any binary choice game ($\Nplayers > 1$ with $J = 2$) with meaningful player-specific states ($K_1 > 1$) is identified, regardless of the number of players or exogenous market states $K_0$.
The sufficient condition in this case simplifies to $K_0 (K_1^\Nplayers - 1) \geq N - \frac{1}{2}$.
When $K_0 \geq 1$ and $K_1 \geq 2$ we have $K_0 (K_1^\Nplayers - 1) \geq 2^\Nplayers - 1$ which exceeds $\Nplayers - \frac{1}{2}$ for integers $\Nplayers > 1$.

\subsubsection{Identification of $Q_i$}

Once the $Q$ matrix is known---or in the case of continuous-time data, identified directly---we need to ensure that in any particular state it does not represent a mixture over potentially multiple equilibria.
To guarantee this, we invoke an assumption corresponding to Assumption 6 of ABBE, which was in turn a continuous-time version a similar assumption required for identification and estimation of discrete time dynamic games \citep{bajari-benkard-levin-2007, aguirregabiria-mira-2007}. See \cite{aguirregabiria-mira-2010} for a survey.

\begin{assumption}[Multiple Equilibria]
  \label{assp:eq}
  The continuous time data generating process is such that in each state $k = 1, \dots, K$:
  (a) A single Markov perfect equilibrium is played corresponding with row $k$ of the intensity matrix $Q$.
  (b) Players' expectations about the distribution of state transitions are consistent with the intensity matrix $Q$.
\end{assumption}

In a model with a unique equilibrium---for example the single agent renewal model---this assumption is satisfied trivially.
In games, it requires that in any markets where the game is in the same state, the same equilibrium is played.
We need this assumption no matter if we observe discrete time data from $P(\Delta)$, generated from some continuous-time $Q$, or we observe continuous-time data generated from $Q$ directly.

Next, we make the following assumption which requires that given the aggregate intensity matrix $Q$, we can determine the player-specific intensity matrices $Q_i$.

\begin{assumption}
  \label{assp:Q_i}
  The mapping $Q \to \{ Q_0, Q_1, \dots, Q_\Nplayers \}$ is known.
\end{assumption}

Assumption~\ref{assp:Q_i} can be easily verified by inspecting $Q$ in both running examples since players cannot change each other's state variables and actions by nature can be distinguished from player actions.
Note that the diagonal elements are unimportant: if the off-diagonal elements of each $Q_i$ can be identified from $Q$, then diagonal elements equal the negative sum of the off-diagonal elements.
In the renewal example $Q$ is given in \eqref{eq:rust:Q} and for the two-player entry model in \eqref{eq:2x2x2:Q}.
A sufficient condition for \autoref{assp:Q_i} is that continuation states resulting from actions of different players are distinct:
for all players $i$ and $m \neq i$ and all states $k$,
\begin{equation*}
  \{ l(i,j,k): j = 1, \dots, J-1 \} \cap \{ l(m,j,k): j = 1, \dots, J-1 \} = \varnothing.
\end{equation*}

\subsection{Identification of Hazards, Value Functions and Payoffs}
\label{sec:identification:v}

We now establish that the value functions, instantaneous payoffs, and utility
functions are identified.
Let $V_i = ( V_{i1}, \dots, V_{iK} )^\top$ denote the $K$-vector of valuations
for player $i$ in each state.
Let $\psi_{ij} = ( \psi_{ij1}, \dots, \psi_{ijK} )^\top$ denote the $K$-vector
of instantaneous payoffs for player $i$ making choice $j$ in each state and
let $\psi_i = (\psi_{i1}^\top, \dots, \psi_{i,J-1}^\top)^\top$.

Importantly, we note that when $j = 0$ is a latent or unobserved continuation
action, it is not possible to identify the rates $h_{i0k}$ even with
continuous time data, so we cannot immediately treat them as identified
quantities.
Under Assumption~\ref{assp:tiev}, the relationship between identified
hazards $h_{ijk}$ for $j > 0$ and the unknown quantities
$h_{i0k}$, $\psi_{ijk}$, and $V_{ik}$ takes the form of a log-linear
system.
Let $h_i^+$ denote the vector of identified hazards for choices $j > 0$
and $h_i^0 = (h_{i01}, \dots, h_{i0K})^\top$ denote the vector
of unidentified hazards for the continuation action $j = 0$.
The hazards for choices $j = 1, \dots, J-1$ are identified from $Q$,
but without additional restrictions, the system has $2K$ more
unknowns than equations, preventing identification of the
remaining unknowns.

The following theorem shows that with $2K$ appropriate linear restrictions,
all these quantities can be identified.
Notably, the number of restrictions required per player is independent
of the number of players in the game, so the total number of
identifying restrictions is only linear in $\Nplayers$.
This contrasts with discrete time models where the number of
restrictions needed is exponential in $\Nplayers$
\citep{pesendorfer08asymptotic}.

\begin{theorem}
  \label{thm:ident:v:psi:h0}
  Suppose Assumptions~\ref{assp:dx}--\ref{assp:Q_i} hold.
  Then, for each player $i$ the identified log
  hazards $\ln h_i^+$ form a linear system with the unidentified
  quantities $\ln h_i^0$, $\psi_i$, and $V_i$.
  Augmented with linear restrictions represented by a matrix $R_i$
  and vector $r_i$, the system becomes
  \begin{equation*}
    \begin{bmatrix}
      X_i \\ R_i
    \end{bmatrix}
    \begin{bmatrix} \ln h_i^0 \\ \psi_i \\ V_i \end{bmatrix}
    = \begin{bmatrix} \ln h_i^+ \\ r_i \end{bmatrix},
  \end{equation*}
  where $X_i$ is a known $(J-1)K \times (J+1)K$ matrix
  defined in \eqref{eq:ident:v:psi:h0:X_i}.
  The matrix $X_i$ has rank $(J-1)K$.
  If R contains $2K$ additional full-rank restrictions such that
  $\left[ \begin{smallmatrix} X_i \\ R_i \end{smallmatrix} \right]$
  has rank $(J+1)K$, then $h_i^0$, $\psi_i$, and $V_i$
  are identified.
\end{theorem}

It is helpful now to consider some examples.
If we assume that the instantaneous payoffs are constant across $k$, as in
many applications of dynamic games, this implies
$\psi_{ijk} - \psi_{ijl} = 0$ for all choices $j > 0$ and all states
$l \neq k$, giving $(J-1)(K-1)$ restrictions per player.
When $J=2$, we still need $K+1$ additional restrictions.
If we also assume the move arrival rate is constant across state
($\sum_{j=0}^{J-1} h_{ijk} = \sum_{j=0}^{J-1} h_{ijl}$ for all $l \neq k$),
we have $K-1$ restrictions.
Then even if $J=2$, only $2$ additional restrictions are needed.

Additional full-rank restrictions are possible for certain applications.
Examples include states where the value function is known, e.g., if
$V_{ik} = 0$ when a firm has permanently exited.
Exclusion restrictions of the form $V_{ik} = V_{ik'}$ are possible, where
$k$ and $k'$ are two states that differ only by a rival-specific state and
are payoff equivalent to firm $i$.
In all of these cases, the rank condition can be verified by inspection.
We also note that Theorem~\ref{thm:ident:v:psi:h0} does not consider identification restrictions across players, but in practice these can provide additional identifying restrictions.

\begin{ex:renewal}[continued]
In the single-agent renewal model, since the replacement cost does not depend on the mileage state we have $\psi_{1k} = \mu$ for all $k$.
This yields $K-1$ restrictions of full rank of the form $\psi_{1k} - \psi_{11} = 0$ for $k = 2, \dots, K$.
If we also assume the rate of move arrivals is constant across two subsets of states (i.e., $\lambda_{\text{L}}$ and $\lambda_{\text{H}}$), this yields $K-2$ additional restrictions.
The linearity of the utility function also imposes restrictions on $V$, and although this does not fit in the linear restriction framework of \autoref{thm:ident:v:psi:h0} it also contributes to identification of $\psi$ and $V$.
\end{ex:renewal}

\begin{ex:entry}[continued]
In the simple two-player entry-exit model, we may suppose that the entry costs and scrap values are independent of the market state (high or low demand) and whether a rival is present.
In other words, $\psi_{i1k} - \psi_{i11} = 0$ for all states $k$, yielding $K-1$ restrictions per player.
Additionally, if we assume the rate of move arrivals is firm-specific ($\lambda_{ik} = \lambda_i$), this yields $K-1$ restrictions per player.
Alternatively, we considered rates depending only on the level of demand ($\lambda_{\text{L}}$ and $\lambda_{\text{H}}$).
This specification would yield $K-2$ restrictions.
\end{ex:entry}

Finally, we note that in practice the overall rate of actions can be identified through the nonlinear restrictions imposed by the distributional assumptions on the error term, which imply shape restrictions on the choice probabilities across states.
These are difficult to characterize in the linear restriction framework we have used here, but in practice parametric assumptions will aid identification in addition to the linear restrictions considered above.

\subsection{Identification of the Payoffs}
\label{sec:identification:u}

Having established the identification of value functions, instantaneous
payoffs, and hazards, we now turn to the identification of flow payoffs.
The following theorem shows that these payoffs can be recovered from the
previously identified quantities through the linear representation of
the value function in \eqref{eq:v:linear} established by
Theorem~\ref{thm:v:linear}.

\begin{theorem}[Identification of Flow Payoffs]
  \label{thm:ident:u}
  Suppose Assumptions~\ref{assp:dx}--\ref{assp:Q_i} hold.
  If for any player $i$ the quantities $V_i$, $\psi_i$, and $Q$ are
  identified, then the flow payoffs $u_i$ are also identified.
\end{theorem}

\section{Empirical Example and Monte Carlo Experiments}
\label{sec:mc}

In this section, we describe an empirical example along with a series of Monte Carlo experiments conducted using the single-agent renewal model.
Additional experiments using a dynamic quality ladder model are presented in \autoref{sec:examples:ladder}.

\subsection{Maximum Likelihood Estimation}

The model can be estimated using maximum likelihood if equilibria can be
enumerated or there is a unique equilibrium.
Since this paper focuses on identification, rather than developing a new
estimator, our Monte Carlo experiments use the maximum likelihood estimator
with value function iteration.\footnote{More generally, methods proposed
  for discrete time models, such as the homotopy method
  \citep{borkovsky10homotopy, besanko10learning,
  bajari-hong-krainer-nekipelov-2010-wp} or recursive lexicographical
  search \citep{iskhakov16recursive}, could possibly be adapted to our
  model, but this is beyond the present scope.}
Multiplicity of equilibria is not a concern for the single agent model and
appears not to be an issue in practice for the continuous-time oligopoly
model specifications considered in the appendix, although we have not
proven that there is a unique equilibrium in the latter case.

However, in models with multiple equilibria, this maximum likelihood
procedure, which relies on value function iteration, is potentially
unstable.
In such cases, we recommend two-step estimators that avoid this issue.
ABBE introduced a two-step PML (pseudo maximum likelihood) estimator,
similar to the CCP estimator of
\cite{hotz93conditional} for discrete-time single-agent models.
More recently, \cite{ctnpl} developed the continuous time NPL (CTNPL)
estimator, an iterative estimator following
\cite{aguirregabiria-mira-2007}.
However, these two-step methods assume the rate of move arrivals $\lambda$
is known.
Adapting them to the general case remains an open question for future
research.
Similarly, other estimators for discrete time models such as those
by \cite{aguirregabiria-marcoux-2019} and \cite{dearing-blevins-2024}
could be adapted to the current framework.

We focus on the maximum likelihood estimator for the
empirical example and simulations.
This  allows us to examine the computational properties and how estimates
behave when the sampling frequency of the data
changes without two-step estimation error.

With continuous-time data, we have a sample of $\bar{N}$ tuples $(\tau_n,i_n,a_n,k_n,k_n')$.
Each describes a jump or move where:
$\tau_n$ is the holding time since the previous event,
$i_n$ is the player index ($i_n = 0$ is nature),
$a_n$ is the action taken by player $i_n$,
$k_n$ denotes the state at the time of the event, and
$k_n'$ denotes the state immediately after.
Let $g(\tau; \lambda)$ and $G(\tau; \lambda)$ denote the pdf and cdf of $\Exponential(\lambda)$.
Let $\ell_n(\theta)$ denote the likelihood of observation $n$:
\begin{equation*}
  \ell_n(\theta)
  = \underbrace{g(\tau_n; q(k_n,k_n;\theta))}_{\text{Arrival time}}
  \left[
    \underbrace{\frac{q_0(k_n,k_n;\theta)}{q(k_n,k_n;\theta)}}_{\text{Event is jump}} \cdot
    \underbrace{p(k_n, k_n'; \theta)}_{\text{Transition}}
  \right]^{\1 \lbrace i_n = 0 \rbrace}
  \left[
    \underbrace{\frac{q_\Nplayers(k_n,k_n;\theta)}{q(k_n,k_n;\theta)}}_{\text{Event is move}} \cdot
    \underbrace{\sigma(i_n, a_n, k_n; \theta)}_{\text{CCP}}
  \right]^{\1 \lbrace i_n > 0 \rbrace}.
\end{equation*}
Here, $q(k,k';\theta)$ denotes the absolute value of the $(k,k')$ element of $Q(\theta)$ for given parameters $\theta$.
$q_0(k,k';\theta)$ and $q_\Nplayers(k,k';\theta)$ similarly denote elements of $Q_0$ and $\sum_{i = 1}^\Nplayers Q_i$.
$p(k,k';\theta)$ denotes the probability of a jump from $k$ to $k'$ conditional on a jump occurring.
The full log-likelihood of the sample of $\bar{N}$ observations on $[0,T]$ is
\begin{equation*}
  \ln L_{\bar{N}}^{\text{CT}}(\theta) = \sum_{n=1}^{\bar{N}} \ln \ell_n(\theta) + \ln \left[ 1 - G(T-t_{\bar{N}},q(k_{\bar{N}},k_{\bar{N}};\theta)) \right].
\end{equation*}
The final term is the probability of not observing an event on $(t_{\bar{N}}, T]$.

With discrete-time data sampled at equispaced intervals $\Delta$ our sample consists of a collection of states $\lbrace k_{1}, \dots, k_{\bar{N}} \rbrace$ with $\bar{N}$ observations.
The likelihood function is given by:
\begin{equation*}
\ln L_{\bar{N}}^{\text{DT}}(\theta) = \sum_{n=2}^{\bar{N}} \ln P\left(k_{n-1}, k_{n}; \Delta, \theta \right),
\end{equation*}
where $P(k,l; \Delta, \theta)$ denotes the $(k,l)$ element of the transition matrix induced by $\theta$.

\subsection{Single Agent Renewal Model: Empirical Results}

We consider the single-agent binary choice (bus engine replacement) model.
We first estimate a continuous time version of the model using the same data
\cite{rust87optimal} used to estimate the original discrete time model.
We then use the estimates to calibrate parameters for Monte Carlo experiments.

The model allows for heterogeneity in move arrival rates across states.
We begin with a simpler model with a constant rate, $\lambda_k = 1$ for
all $k$, as in ABBE.
Next, we allow $\lambda$ to vary freely and estimate it.
Finally, we estimate the model with heterogeneous decision rates where
$\lambda_{\text{L}}$ is the rate for
buses with mileage states $k = 1, 2, \dots, \floor{K / 2}$ and
$\lambda_{\text{H}}$ is the rate for mileage states
$k = \floor{K / 2} + 1, \dots, K$.
The parameters to estimate are
$\theta = (\lambda_{\text{L}}, \lambda_{\text{H}}, \gamma, \beta, \mu)$,
which include the move arrival rates,
the rate of mileage increase $\gamma$,
the mileage cost parameter $\beta$,
and the engine replacement cost $\mu$.

The dataset consists of monthly bus mileage recordings and the months of bus
engine replacement.
We provide statistics on the number of observations per bus group in
Table~\ref{tab:mc1p:rustct:data}.
The smallest time horizon per bus was 24 months (2 years) for group 1.
The longest was 125 months ($\approx{}10$ years) for groups 5--9.
As reported by \cite{rust87optimal}, engine replacement occurred on average after 5
years (60 months) at over 200,000 miles.
This provides a basis for interpreting hazard rates in the continuous time model, where
the dataset consists of monthly observations spaced at equal intervals $\Delta = 1$.

\begin{table}[tbh]
  \centering
  \begin{tabular}{lrrrrr}
    \hline
    \hline
    Bus       &       & Months  &  \\
    Group     & Buses & Per Bus & Bus-Months \\
    \hline
    1         &  15 &  24 &    360 \\
    2         &   4 &  48 &    192 \\
    3         &  48 &  69 &  3,312 \\
    4         &  37 & 116 &  4,292 \\
    5         &  12 & 125 &  1,500 \\
    6         &  10 & 125 &  1,250 \\
    7         &  18 & 125 &  2,250 \\
    8         &  18 & 125 &  2,250 \\
    \hline
    Total     &  162 & -- & 15,406 \\
    \hline
  \end{tabular}
  \caption{\cite{rust87optimal} Sample Characteristics}
  \label{tab:mc1p:rustct:data}
\end{table}

We use the full-solution maximum likelihood approach to estimate the
model.
We fixed the discount rate at $\rho = 0.05$ and
the number of mileage states at $K = 90$.
The value functions are obtained through value function iteration for
each value of $\theta$ in an inner loop to within a tolerance of
$\epsilon_V = 10^{-16}$ under the relative supremum norm.
We maximized the likelihood function in an outer loop using the
L-BFGS-B algorithm \citep{byrd95limited, zhu97algorithm} with central
finite difference derivatives with an adaptive step size proportional to $\epsilon_d^{1/3}$,
with $\epsilon_d = 10^{-8}$.
For robustness to local optima, we took the estimates to be the parameter
values which achieved the highest likelihood over 20 random starting
values.

Although this approach is straightforward, it is computationally intensive.
For each iteration, it requires solving the fixed
point problem for each trial value of $\theta$ and for small steps in the
direction of each component of $\theta$. Alternative methods, such as the NFXP
approach by \cite{rust87optimal}, utilize analytical derivatives of the Bellman
operator to compute analytical derivatives of the log-likelihood function.
Combined with the BHHH algorithm \citep{bhhh-1974}, which approximates the
Hessian of the log-likelihood via the outer product of scores using
the information matrix identity, this provides substantial computational
savings in models with many parameters.
Although the Bellman operator in continuous time is
differentiable, this requires computing analytical derivatives of
the matrix exponential with respect to individual components of the
matrix argument.
These methods are computationally expensive, involving
truncation of infinite sums, evaluation of numerical integrals, or
eigenvalue decompositions of high-dimensional matrices
\citep{magnus-pijls-sentana-2021}.
On the other hand, Newton-Kantorovich-type methods could have particular
advantages in the context of continuous time models, where the $Q$ matrix
is typically very sparse, leading to sparse derivatives of the Bellman
operator.

The estimated structural parameters and standard errors
are reported in Table~\ref{tab:mc1p:rustct:estimates}.
The first column of results corresponds to the model where we
hold fixed $\lambda = 1$ (i.e., $\lambda_{\text{H}} = \lambda_{\text{L}} = 1$).
In this model, the manager is assumed to make decisions on average
once per month, corresponding to the timing of decisions in a discrete time model.
The second column contains estimates for the model where we
allow $\lambda$ to vary and estimate it (i.e., $\lambda = \lambda_{\text{H}} = \lambda_{\text{L}}$).
The final column reports estimates for the heterogeneous version
of the model where $\lambda_{\text{H}}$ may differ from $\lambda_{\text{L}}$.

In the variable $\lambda$ model, we can see that the estimated decision rates are quite different from 1.
Therefore, this provides an interesting setting in which to compare the estimated costs and differences in interpretation.
The variable $\lambda$ specification indicates a relatively low rate of monitoring, with $\hat\lambda = 0.032$ (vs. $\lambda = 1$), but a higher cost of mileage, $\hat\beta = -1.257$ (vs. $\hat\beta = -0.533$).
When $\lambda$ is constrained to 1 (forcing monthly decision-making), the model compensates by estimating a lower mileage cost to match the observed replacement timing.
A manager who checked frequently but with high mileage costs would replace the engine more often than observed in the data.

In the heterogeneous model, there appears to be a slight decrease in the estimated rate of monitoring in lower mileage states, with $\hat\lambda_{\text{L}} = 0.022$ for lower mileage states as compared to $\hat\lambda_{\text{H}} = 0.033$ in high mileage states.
The estimated cost of mileage is $\hat\beta = -1.711$ and the cost of replacement is $\hat\mu = -9.643$.

To choose between these three nested specifications, we carry out likelihood ratio tests of the null hypotheses of homogeneity, $H_0: \lambda_{\text{H}}=\lambda_{\text{L}}$, and decision rates on average equal to monthly decisions in the discrete time model, $H_0: \lambda = 1$.
We fail to reject the homogeneity restriction $\lambda_{\text{H}} = \lambda_{\text{L}}$, but strongly reject the specification with $\lambda = 1$.
It appears to be important to let the rate of decisions vary as a parameter to be estimated, but perhaps they are constant across mileage states in this setting.

\begin{table}[tbh]
  \centering
  \begin{tabular}{lrcrcrc}
    \hline
    \hline
    & \multicolumn{2}{c}{Fixed $\lambda = 1$} & \multicolumn{2}{c}{Variable $\lambda$} & \multicolumn{2}{c}{Heterogeneous $\lambda$} \\
                                           & Est.      & S.E.    & Est.      & S.E.    & Est.      & S.E.    \\
    \hline
    Decision rate   ($\lambda$)            & 1.000     & --      & 0.032     & (0.005) & --        & --      \\
    Decision rate 1 ($\lambda_{\text{L}}$)   & --        & --      & --        & --      & 0.022     & (0.004) \\
    Decision rate 2 ($\lambda_{\text{H}}$)   & --        & --      & --        & --      & 0.033     & (0.005) \\
    Mileage increase ($\gamma$)            & 0.526     & (0.006) & 0.526     & (0.006) & 0.526     & (0.006) \\
    Mileage cost ($\beta$)                 & -0.533    & (0.052) & -1.257    & (0.285) & -1.711    & (0.493) \\
    Replacement cost ($\mu$)               & -8.081    & (0.393) & -8.072    & (1.345) & -9.643    & (2.189) \\
    \hline
    Log likelihood                         & -13947.55 &         & -13938.51 &         & -13937.66 &         \\
    Observations                           &     15406 &         &     15406 &         &     15406 &         \\
    \hline
    Test for $H_0: \lambda_{\text{L}} = \lambda_{\text{H}}=1$               &            &         &           &         &           &        \\
    \quad LR                   &         -- &         &     18.08 &         &     19.78 &         \\
    \quad $p$-value           &         -- &         &   0.00002 &         &   0.00005 &         \\
    Test for $H_0: \lambda_{\text{L}} = \lambda_{\text{H}}$                 &            &         &           &         &           &        \\
    \quad LR          &         -- &         &        -- &         &      1.70 &         \\
    \quad $p$-value   &         -- &         &        -- &         &    0.1923 &         \\
    \hline
  \end{tabular}
  \caption{Model Estimates Based on Data from \cite{rust87optimal}}
  \label{tab:mc1p:rustct:estimates}
\end{table}

\subsection{Single Agent Renewal Model: Monte Carlo}

Inspired by the estimates above, we conducted a Monte Carlo experiment
using the model with true parameters specified as follows:
$(\lambda_{\text{L}}, \lambda_{\text{H}}, \gamma, \beta, \mu) = (0.05, 0.10, 0.5, -2.0, -9.0)$.
We also report estimates of the cost ratio $\mu/\beta = 4.5$ which, as is common in discrete choice models, is more precisely
estimated in most specifications than
$\beta$ or $\mu$ individually.

In the Monte Carlo, we estimate the model under several different
sampling regimes including full continuous-time data and discrete time
data sampled at short and long intervals $\Delta = 1$ and
$\Delta = 8$.
Recall that in the real dataset, $\Delta = 1$ corresponds to a time
period of one month.
In the simulation, we can interpret $\Delta = 8$ as observing the
manager's decision only once every 8 months.
We simulate data over a fixed time interval $[0,T]$ with $T = 120$ months for
each of $M$ markets, with $M$ varying from $200$ to $3,200$.
Recall from Table~\ref{tab:mc1p:rustct:data} that the maximum time
horizon was $T = 125$, so our simulation time horizon is slightly
shorter.
Similarly, in the actual dataset we observed $M = 162$ buses.
Our simulated small sample size is $M = 200$, and we increase that
to $M = 800$ and then $M= 3200$ to evaluate the large sample properties
of the estimator.

\begin{table}[tbh]
 \centering
 \caption{Single Agent Renewal Model Monte Carlo Results}
 \label{tab:mc1p:fixed_t:nfxp}
 \begin{tabular}{rllrrrrrr}
   \hline
   \hline
   $M$   & Sampling        &      & $\lambda_{\text{L}}$ & $\lambda_{\text{H}}$ & $\gamma$  & $\beta$ & $\mu$    & $\mu/\beta$ \\
   \hline
   $\infty$ & DGP          & True & 0.050 & 0.100 & 0.500 & -2.000 & -9.000 & 4.500 \\
   \hline
     200     & Continuous      & Mean    &       0.050  &       0.100  &       0.500  &      -2.050  &      -9.178  &       4.503  \\
             &                 & S.D.    &       0.007  &       0.008  &       0.004  &       0.310  &       1.096  &       0.228  \\
     200     & $\Delta = 1.00$ & Mean    &       0.051  &       0.100  &       0.508  &      -2.079  &      -9.235  &       4.467  \\
             &                 & S.D.    &       0.007  &       0.008  &       0.004  &       0.317  &       1.117  &       0.222  \\
     200     & $\Delta = 8.00$ & Mean    &       0.051  &       0.100  &       0.508  &      -2.093  &      -9.284  &       4.471  \\
             &                 & S.D.    &       0.009  &       0.009  &       0.005  &       0.374  &       1.281  &       0.265  \\
     \hline
     800     & Continuous      & Mean    &       0.050  &       0.100  &       0.500  &      -1.988  &      -8.957  &       4.510  \\
             &                 & S.D.    &       0.003  &       0.005  &       0.002  &       0.121  &       0.427  &       0.099  \\
     800     & $\Delta = 1.00$ & Mean    &       0.051  &       0.101  &       0.508  &      -2.011  &      -8.999  &       4.479  \\
             &                 & S.D.    &       0.003  &       0.005  &       0.002  &       0.124  &       0.433  &       0.098  \\
     800     & $\Delta = 8.00$ & Mean    &       0.051  &       0.100  &       0.508  &      -2.018  &      -9.020  &       4.474  \\
             &                 & S.D.    &       0.003  &       0.005  &       0.003  &       0.145  &       0.498  &       0.108  \\
     \hline
    3200     & Continuous      & Mean    &       0.050  &       0.100  &       0.500  &      -1.995  &      -8.999  &       4.511  \\
             &                 & S.D.    &       0.002  &       0.002  &       0.001  &       0.072  &       0.238  &       0.060  \\
    3200     & $\Delta = 1.00$ & Mean    &       0.051  &       0.100  &       0.508  &      -2.014  &      -9.025  &       4.484  \\
             &                 & S.D.    &       0.002  &       0.002  &       0.001  &       0.072  &       0.233  &       0.061  \\
    3200     & $\Delta = 8.00$ & Mean    &       0.051  &       0.100  &       0.508  &      -2.009  &      -9.004  &       4.485  \\
             &                 & S.D.    &       0.002  &       0.002  &       0.001  &       0.075  &       0.244  &       0.064  \\
     \hline
 \end{tabular}                                                              \\
 \medskip
 \begin{footnotesize}
   The mean and standard deviation are reported for
   100 replications under several sampling regimes.
   For each replication, $M$ markets were simulated over a fixed time
   interval $[0,T]$ with $T = 120$.
 \end{footnotesize}
\end{table}

For each specification, we report the mean and standard deviation
of the parameter estimates over 100 replications in \autoref{tab:mc1p:fixed_t:nfxp}.
With the smallest sample size, $M=200$, although the rate parameters
$\lambda_{\text{L}}$, $\lambda_{\text{H}}$, and $\gamma$ are quite precisely
estimated in all cases---even with a long time interval between discrete
time observations---the cost parameters $\beta$ and $\mu$ are overestimated.
However, even still, they are overestimated in a way such that the
ratio $\mu / \beta$ is close to the true value.
In large samples---as we increase the sample size to $M = 800$ and
$M = 3200$---all are parameters are estimated quite precisely and with
little bias.
The loss of precision, measured by the standard deviations of the
parameter estimates, is minimal when moving from continuous-time data
to discrete-time data with $\Delta = 1$, but it becomes more
noticeable at $\Delta = 8$.

\section{Conclusion}
\label{sec:conclusion}

In this paper, we developed new results on the theoretical and
econometric properties of a generalized instance of the empirical
framework introduced by \citet*{abbe-2016} for continuous time dynamic
discrete choice games.
We showed that the rate of move arrivals can be identified, even in models
where it is heterogeneous, depending on the player identity or state,
whereas previously it was assumed to be known.
We established equilibrium existence
with heterogeneous players and state-dependent move arrival rates,
developed conditions for identification
with discrete time data in the more general model,
explored these results in the context of
three canonical examples widely used in applied work,
and examined the computational properties of the model
as well as the finite- and large-sample
properties of estimates through a series of small- and large-scale Monte
Carlo experiments based on familiar models.

\appendix

\section{A Continuous-Time Quality Ladder Model of Oligopoly Dynamics}
\label{sec:examples:ladder}

To illustrate the application to dynamic games used in empirical industrial
organization we consider a discrete control
version of the quality ladder model proposed by \cite{ericson95markov-perfect}.
This model has been examined extensively by
\cite{pakes94computing,pakes01stochastic},
\cite{doraszelski-satterthwaite-2010}, \cite{doraszelski07framework},
and others.
The model consists of at most $\Nplayers$ firms who compete in a single product
market.
The products are differentiated in that the product of firm $i$ at time $t$
has some quality level $\omega_{it} \in \Omega$, where
$\Omega = \lbrace 1, 2, \dotsc, \bar\omega, \bar\omega +1 \rbrace$
is the finite set of possible quality levels, with $\bar\omega + 1$
being a special state for inactive firms.
Firms with $\omega_{it} < \bar\omega + 1$ are incumbents.
In contrast to \cite{pakes94computing}, all controls here are discrete:
given a move arrival, firms choose whether or not to invest to move up the
quality ladder, rather than how much to spend to increase their chances of doing so.

We consider the particular example of price competition with a single
differentiated product where firms make entry, exit, and investment
decisions, however, the quality ladder framework is quite general and
can be easily adapted to other settings.
For example,
\cite{doraszelski07advertising}
use this framework in a model of advertising where, as above,
firms compete in a differentiated product market by setting prices,
but where the state $\omega_{it}$ is the share of consumers who are
aware of firm $i$'s product.
\cite{gowrisankaran99dynamic} develops a model of endogenous
horizontal mergers where $\omega_{it}$ is a capacity level
and the product market stage game is Cournot with a given demand
curve and cost functions that enforce capacity constraints
depending on each firm's $\omega_{it}$.

To allow for firm and state heterogeneity in move arrival rates,
we may think that some firms monitor the market more frequently in
some states than others, and thus have a higher move arrival rate
$\lambda_{ik}$.
We will suppose that the frequency of monitoring is related to the
quality of the firm's product.
We assume that firms with endogenously higher product quality monitor
the market more frequently than those with lower product quality
and/or potential entrants.
We define ``high product quality'' as $\omega_{it} \geq \omega^{\text{h}}$.
Therefore, we assume that
$\lambda_{ik} = \lambda_{\text{L}}$ for incumbents with $\omega_{it} < \omega^{\text{h}}$
and for potential entrants with $\omega_{it} = \bar\omega + 1$,
while $\lambda_{ik} = \lambda_{\text{H}}$
for incumbents with $\omega_{it} \geq \omega^{\text{h}}$.

\subsection{State Space Representation}

We make the usual assumption that firms are symmetric and
anonymous.
That is, the primitives of the model are the same for each
firm and only the distribution of firms across states,
not the identities of those firms, is payoff-relevant.
By imposing symmetry and anonymity, the size of the state space can be reduced
from the total number of distinct market structures,
$(\bar\omega + 1)^\Nplayers$,
to the number of possible distributions of $\Nplayers$ firms
across $\bar\omega + 1$ states.\footnote{
In practice, we use the ``probability density space'' encoding algorithm
described in \cite{gowrisankaran99efficient}, to map market structure
tuples $s \in \mathcal{S}$ to integers $x \in \mathcal{X}$.}
The set of relevant market configurations is thus the set of ordered
tuples of length $\bar\omega + 1$ whose elements sum to
$\Nplayers$, denoted
$\mathcal{S} = \{ (s_1,\dotsc,s_{\bar\omega + 1}): \sum_j s_j = \Nplayers, s_j \in \Z^{*} \}$,
where $\Z^{*}$ is the set of nonnegative integers.
In this notation, each vector
$\omega = (\omega_1, \dots, \omega_\Nplayers) \in \Omega^\Nplayers$
maps to an element
$s = (s_1,\dots,s_{\bar\omega+1}) \in \mathcal{S}$
with $s_j = \sum_{i=1}^\Nplayers \1\lbrace \omega_i = j \rbrace$
for each $j$.

Each firm also needs to track its own quality, so
payoff relevant market configurations from the perspective of firm $i$
are described by a tuple $(\omega_i,s) \in \Omega \times \mathcal{S}$,
where $\omega_i$ is firm $i$'s quality level and
$s$ is the market configuration.
For our implementation, we map the multidimensional space $\Omega \times \mathcal{S}$
to an equivalent one-dimensional state space
$\mathcal{K} = \lbrace 1, \dotsc, \abs{\Omega} \times \abs{\mathcal{S}} \rbrace$
so that we can represent quantities in matrix-vector form
and we use pre-computed transition addresses to avoid
re-computing continuation states.

\subsection{Product Market Competition}

Again, we follow \cite{pakes94computing} in assuming a continuum of
consumers with measure $\bar{M} > 0$ and that each consumer's utility from
choosing the good produced by firm $i$ is
$g(\omega_i) - p_i + \varepsilon_i$, where $\varepsilon_i$
is iid across firms and consumers and follows a type I extreme value
distribution.
The $g$ function is used to enforce an upper bound on profits.
As in \cite{pakes93implementing}, for some constant $\omega^*$
we define
\begin{equation*}
 g(\omega_i) = \begin{cases}
   \omega_i & \text{if } \omega_i \leq \omega^*, \\
   \omega_i - \ln(2 - \exp(\omega^* - \omega_i)) & \text{if } \omega_i > \omega^*.
 \end{cases}
\end{equation*}

Let $\mathcal{s}_i(\omega, p)$ denote firm $i$'s market share given the
state $\omega$ and prices $p$.
From \cite{mcfadden74conditional}, we know that the share of consumers
purchasing good $i$ is
\begin{equation*}
 \mathcal{s}_i(\omega, p) =
 \frac{\exp(g(\omega_i) - p_i)}%
 {1 + \sum_{j=1}^\Nplayers \exp(g(\omega_j) - p_j)}.
\end{equation*}
In a market of size $\bar{M}$, firm $i$'s demand is
$q_i(\omega, p) = \bar{M} \mathcal{s}_i$.

All firms have the same constant marginal cost $c \geq 0$.
Taking the prices of other firms, $p_{-i}$, as given,
the profit maximization problem of firm $i$ is
\begin{equation*}
 \max_{p_i \geq 0} q_i(p, \omega) (p_i - c).
\end{equation*}
\cite{caplin91aggregation} show that (in this single-product firm setting)
there is a unique Bertrand-Nash
equilibrium, which is given by the solution to the first order conditions:
\begin{equation*}
 \frac{\partial q_i}{\partial p_i}(p, \omega) (p_i - c)
 + q_i(p, \omega) = 0.
\end{equation*}
Given the functional forms above, the first order conditions become
\begin{equation*}
 -(p_j - c)(1 - \mathcal{s}_j) + 1 = 0.
\end{equation*}
We solve this nonlinear system of equations numerically using the
Newton-Raphson algorithm to obtain the equilibrium prices and
the implied profits
$\pi(\omega_i, \omega_{-i}) = q_i(p, \omega) (p_i - c)$
earned by each firm $i$ in each state $(\omega_i, \omega_{-i})$.

\subsection{Incumbent Firms}

We consider a simple model in which incumbent firms have three
choices.
Firms may continue without investing at no cost,
they may invest an amount $\kappa$ in order to increase the quality of
their product from $\omega_i$ to
$\omega_i' = \min \lbrace \omega_i + 1, \bar\omega \rbrace$,
or they may exit the market and receive some scrap value $\phi$.
We denote these choices, respectively, by the choice set
$\mathcal{J} = \lbrace 0, 1, 2 \rbrace$.
When an incumbent firm exits the market,
$\omega_i$ jumps deterministically to $\bar\omega + 1$.
Associated with each choice $j$ is a private shock
$\varepsilon_{ijk}$.
These shocks are iid over firms, states, choices, and time and
follow a standard type I extreme value distribution
(Assumption~\ref{assp:tiev}).

For any market-wide state $k \in \mathcal{K}$, let
$\omega_k = (\omega_{k1}, \dotsc, \omega_{k\Nplayers})$ denote
the corresponding market configuration in $\Omega^\Nplayers$.
In the general notation introduced above, the instantaneous payoff
$\psi_{ijk}$ to an incumbent firm $i$ from choosing choice $j$ in state $k$ is
\begin{equation*}
 \psi_{ijk} = \begin{cases}
   0 & \text{if } j = 0, \\
   -\kappa & \text{if } j = 1, \\
   \phi & \text{if } j = 2.
 \end{cases}
\end{equation*}

The state resulting from continuing ($j = 0$) is simply
$l(i,0,k) = k$.
Similarly, for investment ($j = 1$),
$l(i,1,k) = k'$
where state $k'$ is the element of $\mathcal{X}$ such that
$\omega_{k'i} = \min\{ \omega_{ki} + 1, \bar\omega \}$ and
$\omega_{k'm} = \omega_{km}$ for all firms $m \neq i$.
Note that we are considering only incumbent firms with
$\omega_{ki} < \bar\omega + 1$.
Exiting is a terminal action with an instantaneous payoff but no
continuation value.

Each incumbent firm pays a constant flow fixed cost $\mu$ while
remaining in the market, and receives the flow profits $\pi_{ik} = \pi(\omega_{ki}, \omega_{k,-i})$
associated with product market competition.
The value function for an incumbent firm in state $k$ is thus
\begin{multline*}
 V_{ik} = \frac{1}{\rho + \sum_{l \neq k} q_{kl} + \sum_{m=1}^{\Nplayers} \lambda_{mk}}\left(
   \pi_{ik} - \mu
   + \sum_{l \neq k} q_{kl} V_{il}
   + \sum_{m \neq i} \lambda_{mk} \sum_j \sigma_{mjk} V_{i,l(m,j,k)}
   \right. \\ \left.
   + \lambda_{ik} \E \max\left\{
       V_{ik} + \varepsilon_{i0k},
       V_{i,l(i,1,k)} - \kappa + \varepsilon_{i1k},
       \phi + \varepsilon_{i2}
     \right\}
\right).
\end{multline*}
Conditional upon moving while in state $k$, incumbent firms face the
maximization problem
$\max \left\lbrace
   V_{ik} + \varepsilon_{i0},
   -\kappa + V_{ik'} + \varepsilon_{i1},
   \phi + \varepsilon_{i2}
 \right\rbrace.$
The resulting choice probabilities are
\begin{gather*}
 \sigma_{i0k} = \frac{\exp(V_{ik})}{\exp(V_{ik}) + \exp(-\kappa + V_{ik'}) + \exp(\phi)}, \\
 \sigma_{i1k} = \frac{\exp(-\kappa + V_{ik'})}{\exp(V_{ik}) + \exp(-\kappa + V_{ik'}) + \exp(\phi)}, \\
 \sigma_{i2k} = 1 - \sigma_{i0k} - \sigma_{i1k},
\end{gather*}
where, as before, $k' = l(i,1,k)$ denotes the resulting state after
investment by firm $i$.

\subsection{Potential Entrants}

Whenever the number of incumbents is smaller than $\Nplayers$,
a single potential entrant receives the opportunity to enter at rate
$\lambda_{\text{L}}$.
Potential entrants are short-lived and do not consider the option
value of delaying entry.
If firm $i$ is a potential entrant with the opportunity to move it has two
choices: it can choose to enter ($j = 1$), paying a setup cost
$\eta$ and entering the market immediately in a predetermined
entry state $\omega^e \in \Omega$
or it can choose not to enter ($j = 0$) at no cost.
Associated with each choice $j$ is a stochastic private payoff shock
$\varepsilon_{ijk}$.
These shocks are iid across firms, choices, and time, and are
distributed according to the type I extreme value distribution
(Assumption~\ref{assp:tiev}).

In our general notation,
for actual entrants
($j = 1$) in state $k$ the instantaneous payoff is
$\psi_{i1k} = -\eta$ and the continuation state is
$l(i,1,k) = k'$ where $k'$ is the element of $\mathcal{K}$ with
$\omega_{k'i} = \omega^e$
and $\omega_{k'm} = \omega_{km}$ for all $m \neq i$.
For firms that choose not to enter ($j = 0$) in state $k$, we have
$\psi_{i0k} = 0$ and the firm leaves the market with no continuation
value.
Thus, upon moving in state $k$, a potential entrant
faces the problem
\begin{equation*}
 \max\left\lbrace
   \varepsilon_{i0k},
   -\eta + V_{ik'} + \varepsilon_{i1k}
 \right\rbrace
\end{equation*}
yielding the conditional entry-choice probabilities
\begin{equation*}
 \sigma_{i1k} = \frac{\exp(V_{ik'} - \eta)}{1 + \exp(V_{ik'} - \eta)}.
\end{equation*}

\subsection{State Transitions}

In addition to state transitions resulting directly from
entry, exit, or investment decisions, the overall state of the market
follows a Markov jump process.
At rate $\gamma$, the quality of each firm $i$ jumps from
$\omega_i$ to $\omega_i' = \max\lbrace \omega_i - 1, 1 \rbrace$.
This represents an industry-wide negative demand shock,
which can be interpreted as an improvement in the outside alternative.

\subsection{Monte Carlo Experiments}

To complement the main Monte Carlo experiments in Section~\ref{sec:mc},
here we carry out another set of simulations using the quality ladder model
described above.
Table~\ref{tab:timing} provides an overview of the model
specifications.
The table covers models with player counts ranging from $\Nplayers =
2$ with $K=56$ states to $\Nplayers = 30$ with $K=58,433,760$ states.
We keep the number of possible quality levels fixed at $\bar{\omega} = 7$.
For simplicity, the quality level threshold for the decision rate is
set to match the entry-level quality, with $\omega^{\text{h}} = \omega^e = 4$.
As the number of potential players ($\Nplayers$) increases, we adjust
the market size ($\bar{M}$) to ensure that the average number of
active players ($n_{\text{avg}}$) grows accordingly.
Additionally, we report $K$, the number of distinct $(\omega_i,\omega)$
state combinations in $\mathcal{X}$, from the perspective of player $i$.

\begin{table}[htbp]
 \centering
 \caption{Quality Ladder Model Monte Carlo Specifications}
 \label{tab:timing}
 \begin{tabular}{rrrrr}
   \hline
   \hline
   $\Nplayers$ & $\bar\omega$ & $K$        & $\bar{M}$   & Obtain $V$ \\
   \hline
    2          & 7            & 56         & 0.40        & 0.15       \\
    4          & 7            & 840        & 0.60        & 0.27       \\
    6          & 7            & 5,544      & 0.75        & 0.65       \\
    8          & 7            & 24,024     & 0.85        & 3          \\
   10          & 7            & 80,080     & 0.95        & 10         \\
   12          & 7            & 222,768    & 1.05        & 30         \\
   14          & 7            & 542,640    & 1.15        & 79         \\
   16          & 7            & 1,193,808  & 1.20        & 199        \\
   18          & 7            & 2,422,728  & 1.25        & 422        \\
   20          & 7            & 4,604,600  & 1.30        & 882        \\
   22          & 7            & 8,288,280  & 1.35        & 1648       \\
   24          & 7            & 14,250,600 & 1.40        & 2964       \\
   26          & 7            & 23,560,992 & 1.45        & 6481       \\
   28          & 7            & 37,657,312 & 1.50        & 10804      \\
   30          & 7            & 58,433,760 & 1.55        & 17712      \\
   \hline
 \end{tabular} \\
 \medskip
 \begin{footnotesize}
   $\Nplayers$ denotes the number of players (including potential entrants),
   $\bar\omega$ denotes the number of quality levels,
   $K$ denotes the total number of distinct states,
   $\bar{M}$ denotes the market size, and
   ``Obtain $V$'' denotes the time in seconds required for value iteration convergence.
   Computational times are wall clock times using GNU Fortran 12.2 on a 2019 Mac Pro with a 2.5 GHz 28-Core Intel Xeon W processor.
 \end{footnotesize}
\end{table}

The final column of Table~\ref{tab:timing} compares the computational
time required (wall clock time) for obtaining the value function
across specifications.
This step is necessary to either generate a dataset or to simulate the
model (e.g., to perform counterfactuals).
We used value function iteration where the stopping criterion is that
the choice probabilities have converged to within a tolerance of
$\epsilon_\sigma = 10^{-8}$ in the supremum norm.

To put the computational times in perspective,
\cite{doraszelski-judd-2012} noted that it would take about \emph{one year} to
just solve for an equilibrium of a comparable\footnote{The times
  reported by \cite{doraszelski-judd-2012} are for a model with
  $\bar{\omega} = 9$ but with no entry or exit, which for a fixed
  value of $\Nplayers$, is roughly comparable in terms of dimensionality
  to our model with $\bar{\omega} = 7$, which includes entry and exit.}
14-player game using the Pakes-McGuire algorithm.
Similar computational times are reported in \cite{doraszelski07framework}.
However, it takes just over \emph{one minute} to solve the continuous-time game
with 14 players and 542,640 states.
Even in the game with 30 players and over 58 million states,
obtaining the value function took under 5 hours.
We note that this would still render full solution estimation infeasible,
but when estimating the model using two-step methods such as in ABBE
or \cite{ctnpl}, one may only need to carry out this step once, after
estimation, for simulating a counterfactual.
Overall, these computational times suggest that a much larger class of
problems can be estimated and simulated in the continuous-time
framework.

\autoref{tab:ladder:nfxp} summarizes the results of our Monte Carlo
experiments.
We estimate the structural parameters
$(\lambda_{\text{L}}, \lambda_{\text{H}}, \gamma, \kappa, \eta, \mu)$.
The true parameter values, which are also shown in the table, are
$(\lambda_{\text{L}}, \lambda_{\text{H}}, \gamma, \kappa, \eta, \mu) = (1.0, 1.2, 0.4, 0.8, 4.0, 0.9)$.
Because we estimate firm fixed costs $\mu$, we set the scrap value received
upon exit equal to zero $(\phi = 0)$.

\begin{table}[tbp]
 \centering
 \caption{Quality Ladder Model Monte Carlo Results}
 \label{tab:ladder:nfxp}
 \begin{tabular}{rrllrrrrrr}
   \hline
   \hline
   $\Nplayers$ & $K$      & Sampling       &      & $\lambda_{\text{L}}$ & $\lambda_{\text{H}}$ & $\gamma$ & $\kappa$ & $\eta$ & $\mu$ \\
   \hline
               &          & DGP            & True & 1.000 & 1.200 & 0.400 & 0.800 & 4.000 & 0.900 \\
   \hline
    2          &       56 & Continuous     & Mean & 0.997 & 1.196 & 0.400 & 0.796 & 3.988 & 0.899 \\
               &          &                & S.D. & 0.015 & 0.020 & 0.010 & 0.032 & 0.137 & 0.021 \\
               &          & $\Delta = 1.0$ & Mean & 1.021 & 1.223 & 0.399 & 0.801 & 3.932 & 0.914 \\
               &          &                & S.D. & 0.177 & 0.181 & 0.007 & 0.283 & 0.841 & 0.063 \\
    4          &      840 & Continuous     & Mean & 0.999 & 1.199 & 0.397 & 0.806 & 4.030 & 0.897 \\
               &          &                & S.D. & 0.013 & 0.018 & 0.014 & 0.033 & 0.160 & 0.022 \\
               &          & $\Delta = 1.0$ & Mean & 0.998 & 1.197 & 0.400 & 0.781 & 3.948 & 0.902 \\
               &          &                & S.D. & 0.114 & 0.113 & 0.006 & 0.180 & 0.456 & 0.040 \\
    6          &    5,544 & Continuous     & Mean & 1.001 & 1.198 & 0.399 & 0.798 & 4.012 & 0.900 \\
               &          &                & S.D. & 0.014 & 0.018 & 0.016 & 0.035 & 0.144 & 0.021 \\
               &          & $\Delta = 1.0$ & Mean & 1.004 & 1.207 & 0.399 & 0.805 & 4.017 & 0.901 \\
               &          &                & S.D. & 0.087 & 0.088 & 0.006 & 0.135 & 0.330 & 0.032 \\
    8          &   24,024 & Continuous     & Mean & 1.000 & 1.201 & 0.400 & 0.802 & 4.027 & 0.899 \\
               &          &                & S.D. & 0.013 & 0.017 & 0.018 & 0.033 & 0.149 & 0.023 \\
               &          & $\Delta = 1.0$ & Mean & 1.012 & 1.213 & 0.400 & 0.814 & 4.030 & 0.905 \\
               &          &                & S.D. & 0.082 & 0.083 & 0.005 & 0.125 & 0.292 & 0.030 \\
   \hline
 \end{tabular}
\end{table}

We first used samples containing $\bar{N} = 10,000$ continuous time events.
In this case, we observe the time of each event, the identity of the
player, and the action chosen.
For each specification, we also report results for
estimation with discrete time data with a fixed sampling interval of
$\Delta = 1$ and $\bar{N} = 10,000$ observations.
In this case, we must calculate the matrix exponential of the $Q$
matrix at each trial value of $\theta$.
To do so, we use the uniformization algorithm as described in \cite{sherlock21direct}.
Because this matrix is high dimensional, but sparse, we
adapted the algorithm to use sparse matrix methods, and we precomputed
the locations of the non-zero elements to improve the computational
speed.

For each replication, we used L-BFGS-B \citep{byrd95limited, zhu97algorithm}
to maximize the log-likelihood function and used
$\epsilon_\sigma = 10^{-13}$ as the tolerance for value function iteration,
checking convergence using the sup norm of the choice probabilities.%
\footnote{Because of the time required to complete many replications of each
  specification, and because the specification has undergone a revision, we
  have limited our consideration to models up to
  $\Nplayers = 8$ players and $K = 24,024$ states
  for the Monte Carlo experiments.}
As before, we use central finite difference derivatives with an adaptive
step size proportional to $\epsilon_d^{1/3}$, with $\epsilon_d = 10^{-8}$.

We took the estimates to be the parameter values which
achieved the highest likelihood over 3 distinct starting values.
Each replication involves a parameter search and each
parameter evaluation solves a full solution problem for
accuracy.\footnote{To ease the computational burden, we store up to
  100 previous value functions and associated parameter values. Then
  for each trial value of $\theta$, we search for the closest (in
  Euclidean distance) previous parameter values and use the associated
  value function as the starting value for value function iteration.}
Although this is computationally costly, it allows us to focus on
identification, computation, and estimation under time aggregation in
a setting without additional tuning parameters and two-step
estimation error.

The estimates are reasonably accurate and precise in all specifications,
including the firm heterogeneity in move arrival rates.
As expected, we can see that the precision is decreased (standard
errors are increased) in most cases due to the information lost with
only discretely sampled data.
Although the standard errors are larger than those with continuous time
data, they are still reasonably small.

\section{Proofs}
\label{sec:proofs}

\begin{proof}[Proof of \autoref{thm:existence}]

First, note that the best response condition in \eqref{eq:best_response}
is equivalent to the following inequality condition:
\begin{equation}
  \label{prule}
  \delta_i(k,\varepsilon_{ik}, \sigma_i) = j
  \iff
  \psi_{ijk} + \varepsilon_{ijk} + V_{i,l(i,j,k)}(\sigma_i)
  \geq
  \psi_{ij'k} + \varepsilon_{ij'k} + V_{i,l(i,j',k)}(\sigma_i)
  \quad \forall j' \in \mathcal{J}.
\end{equation}
Define the mapping
$\Upsilon: [0,1]^{\Nplayers \times J \times K} \to [0,1]^{\Nplayers \times J \times K}$
by stacking best response probabilities:
\begin{equation*}
  \Upsilon_{ijk}(\sigma) = \int \1\left\{
    \varepsilon_{ij'k} - \varepsilon_{ijk} \leq
    \psi_{ijk} - \psi_{ij'k} + V_{i,l(i,j,k)}(\sigma_{-i}) - V_{i,l(i,j',k)}(\sigma_{-i})
    \quad \forall j' \in \mathcal{J}_i
  \right\}\,f(\varepsilon_{ik})\,d\varepsilon_{ik}.
\end{equation*}
$\Upsilon$ is a continuous function from a compact space
onto itself, so By Brouwer's theorem, it has a fixed point.
The fixed point probabilities imply stationary Markov strategies
that constitute a Markov perfect equilibrium.
\end{proof}

\begin{proof}[Proof of \autoref{thm:v:linear}]
Given a collection of equilibrium best response probabilities
$\{ \sigma_i \}_{i=1}^\Nplayers$, we arrived at the linear operator
for the value function $V_i(\sigma_i)$ in \eqref{eq:bellman:matrix}.
As noted, Lemma~\ref{lem:ccp:inversion} guarantees that the difference
$V_{i,l(i,j,k)}(\sigma_i) - V_{i,l(i,j',k)}(\sigma_i)$
can be expressed as a function of payoffs and choice
probabilities $\sigma_i$ and so we can write $C_i$ as a function
of only conditional choice probabilities and payoffs (i.e.,
it no longer depends on the value function).

Noting that $V_i = \Gamma_i(V_i)$ and restating \eqref{eq:bellman:matrix}
to collect terms involving $V_i(\sigma_i)$ yields
\begin{equation*}
  V_i(\sigma_i) \left[ \rho_i I_k + \sum_{m=1}^\Nplayers L_m [I_K - \Sigma_m(\sigma_m)] - Q_0 \right]
  = u_i + L_i C_i(\sigma_i).
\end{equation*}
The matrix in square brackets side is strictly diagonally dominant:
for each $m$
$\rho_m > 0$ by \autoref{assp:rho},
$L_m$ is a diagonal matrix with strictly positive elements by \autoref{assp:rates},
$\Sigma_m(\sigma_m)$ has elements in $[0,1]$ with row sums equal to one,
and elements of $Q_0$ satisfy $\abs{q_{kk}} = \sum_{l \neq k} \abs{q_{kl}}$ in each row $k$.
Therefore, by the Levy-Desplanques theorem \citep[Theorem 6.1.10]{horn-johnson-1985}
this matrix is nonsingular.
\end{proof}

\begin{proof}[Proof of \autoref{thm:id_Q}]
  To establish generic identification of $Q$ we specialize the proof of
  Theorem 1 of \cite{blevins-2017} to the present setting.
  In this setting, $P(\Delta)$ is identified and is the solution to the
  Kolmogorov forward equations while $Q$ is a matrix of unknown parameters
  with $q_{kl}$ for $l \neq k$ being the hazard of jumps from state $k$
  to state $l$.
  The unique solution to this system is the transition matrix
  $P(\Delta) = \exp(\Delta Q)$, which has the same form as the matrix
  $B$ in equation (3) of \cite{blevins-2017} and $Q$ in this model is
  analogous to $A$ in (1).

  For the $K \times K$ matrix $Q = (q_{kl})$, we denote $\vectorize(Q)$
  as the vector obtained by stacking the columns of $Q$:
  $\vectorize(Q) = ( q_{11}, q_{21}, \dots, q_{K1}, \dots, q_{1K}, \dots, q_{KK} )^\top$.
  \cite{gantmacher-1959} showed that all solutions $\tilde{Q}$ to
  $\exp(\Delta \tilde{Q}) = P(\Delta)$ have the form
  $\tilde{Q} = Q + U D U^{-1}$,
  where $U$ is a matrix whose columns are the eigenvectors of $Q$ and
  $D$ is a diagonal matrix containing differences in the complex
  eigenvalues of $Q$ and $\tilde{Q}$.
  This means that both the eigenvectors $U$ and the real eigenvalues
  of $Q$ are identified.

  Any other such matrices $\tilde{Q}$ must also satisfy the prior
  restrictions, so $R \vectorize(\tilde{Q}) = r$.
  By the relationship between $Q$ and $\tilde{Q}$ above, we have
  $R \vectorize(Q + U D U^{-1}) = r$.
  But $R \vectorize(Q) = r$ and by linearity of the vectorization operator,
  $R \vectorize(U D U^{-1}) = 0$.
  An equivalent representation is
  $R (U^{-\top} \otimes U) \vectorize(D) = 0$.

  Since $Q$ is an intensity matrix with row sums equal to zero, it has one
  real eigenvalue equal to zero and at most $K-1$ complex eigenvalues.
  The vector of ones is a right eigenvector of $Q$ with zero as the
  eigenvalue.
  In this case, the number of required restrictions on $Q$ is reduced to
   $\floor{(K-1)/2}$ because we know $Q$ has at least one real eigenvalue.
  When there are at least $\floor{(K-1)/2}$ linear restrictions and $R$
  has full rank, then $D$ must be generically zero and therefore the
  eigenvalues of $\tilde{Q}$ and $Q$ are equal.
  If the eigenvectors and all eigenvalues of $\tilde{Q}$ are the same as
  those of $Q$, the matrices must be equal and therefore $Q$ is identified.

  Under the assumptions the number of distinct states in the model is $K \equiv K_0 K_1^\Nplayers$.
  Therefore, we will require at least $\floor{(K-1)/2}$ linear restrictions of the form $R \vectorize(Q) = r$ where $R$ has full rank.
  We proceed by showing that the present model admits an intensity matrix $Q$ with a known sparsity pattern and so we can use the locations of zeros as homogeneous restrictions, where $r$ will be a vector of zeros.

  Recall that each player has $J$ choices, but $j = 0$ is a continuation choice.
  This results in $J-1$ non-zero off-diagonal elements per row of $Q$ per player.
  There are at most $K_0-1$ non-zero off-diagonal elements due to exogenous state changes by nature.
  The only other non-zero elements of each row are the diagonal elements and therefore there are at least $K - \Nplayers(J-1) - (K_0 - 1) - 1 = K_0 K_1^\Nplayers - \Nplayers(J-1) - K_0$ zeros per row of $Q$.
  The order condition is that the \emph{total} number of zero restrictions is at least $\floor{(K-1)/2}$.
  For simplicity, it will suffice to show that there are $K/2 \geq \floor{(K-1)/2}$ restrictions.
  Summing across rows, this condition is satisfied when
  $(K_0 K_1^\Nplayers)(K_0 K_1^\Nplayers - \Nplayers (J-1) - K_0) \geq K_0 K_1^\Nplayers/2$.
  Simplifying yields the sufficient condition in \eqref{eq:id_Q}.

  In terms of the restrictions required by Theorem 1 of
  \cite{blevins-2017}, the restrictions we have generated all involve
  single-element zero restrictions on $\vectorize(Q)$ in distinct
  locations, therefore the restriction matrix has full rank.

  The derivative of the left-hand-side of \eqref{eq:id_Q} with respect to $K_0$ is $K_1^N - 1$.
  This value is always non-negative, since $K_1 \geq 1$, and is strictly positive when $K_1 > 1$.
  The derivative with respect to $K_1$ is $\Nplayers K_0 K_1^{\Nplayers-1}$.
  This value is always strictly positive since $K_0 \geq 1$ and $K_1 \geq 1$.
  Finally, the derivative with respect to $J$ is $-\Nplayers$.
\end{proof}

\begin{proof}[Proof of \autoref{thm:ident:v:psi:h0}]
Under Assumption~\ref{assp:tiev}, note that
$h_{ijk} = \lambda_{ik} \sigma_{ijk}$ and, recalling the choice probabilities
in \eqref{eq:ccp:logit}, differences in log hazards can be
written as
\begin{equation*}
  \ln h_{ijk} - \ln h_{i0k} = \ln \sigma_{ijk} - \ln \sigma_{i0k} = \psi_{ijk} + V_{i,l(i,j,k)} - V_{ik}.
\end{equation*}
Rearranging, we have
\begin{equation*}
  \ln h_{ijk} = \ln h_{i0k} + \psi_{ijk} + V_{i,l(i,j,k)} - V_{ik}.
\end{equation*}
The hazards on the left hand side for $j = 1, \dots, J-1$ are identified
from $Q$ under Assumption~\ref{assp:Q_i}, while the quantities on the right-hand side are unknowns to be
identified.

Define $S_{ij}$ to be the state transition matrix
induced by the continuation state function $l(i,j,\cdot)$.
In other words, $S_{ij}$ is a permutation matrix where the $(k,l)$ element
is $1$ if playing action $j$ in state $k$ results in a transition to state
$l$ and $0$ otherwise.
Stacking equations across states $k$ and choices $j$ gives:

\begin{equation}
  \label{eq:ident:v:psi:h0:X_i}
  \begin{bmatrix}
    \ln h_{i1} \\
    \vdots \\
    \ln h_{i,J-1} \\
  \end{bmatrix}
  = \left[
    \begin{array}{c|cccc|c}
      I_K    & I_K    & 0      & \dots  & 0      & S_{i1} - I_K \\
      I_K    & 0      & I_K    & \dots  & 0      & S_{i2} - I_K \\
      \vdots & \vdots & \vdots & \vdots & \vdots & \vdots       \\
      I_K    & 0      & 0      & \dots  & I_K    & S_{i,J-1} - I_K \\
    \end{array}
  \right]
  \begin{bmatrix}
    \ln h_{i0} \\
    \hline
    \psi_{i1} \\
    \vdots \\
    \psi_{i,J-1} \\
    \hline
    V_i \\
  \end{bmatrix}.
\end{equation}
Define $X_i$ to be the partitioned matrix above.
This system has $(J-1)K$ equations with $(J+1)K$ unknowns:
$K$ unknown log hazards $\ln h_{i0}$, $(J-1)K$ unknown instantaneous payoffs $\psi_{ij}$,
and $K$ unknown valuations $V_i$.

Under \autoref{assp:dn}, for any action $j > 0$ in any state $k$, the
continuation state is different from $k$.
Therefore, the diagonal elements of $S_{ij}$ are all zero and
$S_{ij} - I_K$ has full rank for each $j > 0$, and these blocks are
linearly independent across $j$.
This means that $X_i$ has rank $(J-1)K$.

The augmented system with $R_i$ and $r_i$ denoting linear restrictions on
the unknowns is:
\begin{equation*}
  \begin{bmatrix} \ln h_i^+ \\ r_i \end{bmatrix}
  =
  \begin{bmatrix}
    X_i \\ R_i
  \end{bmatrix}
  \begin{bmatrix} \ln h_i^0 \\ \psi_i \\ V_i \end{bmatrix}.
\end{equation*}

Since $X_i$ has rank $(J-1)K$, we need $2K$ additional full-rank restrictions
for the augmented matrix to have full column rank $(J+1)K$.
When this rank condition is satisfied, the system has a unique solution
and $h_i^0$, $\psi_i$, and $V_i$ are identified.
\end{proof}

\begin{proof}[Proof of Theorem~\ref{thm:ident:u}]
In light of the linear representation in \eqref{eq:v:linear}, we have
\begin{equation*}
  u_i = \Xi_i(Q) V_i - L_i C_i(\sigma_i),
\end{equation*}
where $\Xi_i(Q)$ is the matrix function defined in \eqref{eq:Xi_i}.

Under the maintained assumptions,
$V_i$, $\psi_i$, and $h_i$ can be identified for each player
by Theorem~\ref{thm:ident:v:psi:h0}.
Since all choice-specific hazards $h_{ijk}$ are identified for all
choices (including $j = 0$),
recalling that $h_{ijk} = \lambda_{ik} \sigma_{ijk}$,
the choice probabilities are also identified:
$\sigma_{ijk} = h_{ijk} / \sum_{j'=0}^{J-1} h_{ij'k}$.
Therefore, all quantities on the right-hand side of the equation
for $u_i$ are known, and $u_i$ can be obtained directly.
\end{proof}

\bibliographystyle{chicago}
\bibliography{ctgames}

\end{document}